\documentclass[acmtosn]{acmtrans2m}

\usepackage{epsfig,graphicx}
\usepackage{amsmath}
\usepackage{amssymb}
\usepackage{array}
\usepackage{url}
\usepackage{theorem}

\newcommand{\mysum}[3]{{\overset {#3}{\underset{#1=#2}\sum}}}  
  
\newcommand{\prob}[1]{\mathsf{P}\left(#1\right)}
\newcommand{\EXP}[1]{\mathsf{E}\!\left[#1\right]}
\newcommand{\MPA}{\min_{\Pi_{\alpha}}}

\newcommand{\nn}{\nonumber \\}
{\theorembodyfont{\rmfamily}
\newtheorem{remarks}{Remark}[section]} 
\newtheorem{theorem}{Theorem}
\newtheorem{lemma}{Lemma}

\graphicspath{{../figures/}{../plots/}}

\title{Delay Optimal Event Detection on Ad Hoc Wireless Sensor
Networks}

\author{K.~Premkumar, Venkata~K.~Prasanthi~M., and Anurag~Kumar\\ \\
        Dept.\ of Electrical Communication Engineering,\\ 
        Indian Institute of Science, Bangalore, INDIA\\ 
        email: kprem@ece.iisc.ernet.in,
               prasanthi.m@gmail.com, 
               anurag@ece.iisc.ernet.in}  

\begin{abstract}
We consider a small extent sensor network for event detection, in which
nodes take samples periodically and then contend over a {\em random
access network} to transmit their measurement packets to the fusion
center. We consider two procedures at the fusion center to process the
measurements. The Bayesian setting is assumed; i.e., the fusion center
has a prior distribution on the change time. In the first procedure, the
decision algorithm at the fusion center is \emph{network--oblivious} and
makes a decision only when a complete vector of measurements taken at a
sampling instant is available. In the second procedure, the decision
algorithm at the fusion center is \emph{network--aware} and processes
measurements as they arrive, but in a time causal order. In this case,
the decision statistic depends on the network delays as well, whereas in
the network--oblivious case, the decision statistic does not depend on
the network delays. This yields a Bayesian change detection problem with
a tradeoff between the random network delay and the decision delay; a
higher sampling rate reduces the decision delay but increases the random
access delay. Under periodic sampling, in the network--oblivious case,
the structure of the optimal stopping rule is the same as that without
the network, and the optimal change detection delay decouples into the
network delay and the optimal decision delay without the network. In the
network--aware case, the optimal stopping problem is analysed as a
partially observable Markov decision process, in which the states of the
queues and delays in the network need to be maintained. A sufficient
statistic for decision is found to be the network--state and the
posterior probability of change having occurred given the measurements
received and the state of the network. The optimal regimes are studied
using simulation.
   
\end{abstract}





\begin{document}

\setcounter{page}{1}

\begin{bottomstuff}
This is an expanded version of a paper that was presented in IEEE
SECON 2006. This work was supported in part by grant number 2900 
IT from the Indo-French Center for the Promotion of Advanced Research 
(IFCPAR), and in part by a project from DRDO, Government of India.
\end{bottomstuff}

\maketitle

\section{Introduction}
\label{sec:introduction}

A wireless sensor network is formed by tiny, untethered devices
(``motes'') that can sense, compute and communicate. Sensor networks 
have a wide range of applications such as environment monitoring, 
detecting events, identifying locations of survivors in building 
and train disasters, and intrusion detection for defense and security 
applications. For factory and building automation applications, there 
is increasing interest in replacing wireline sensor networks with 
wireless sensor networks, due to the potential reduction in costs 
of engineering, installation, operations, and 
maintenance~\cite{honeywell-site}~\cite{isa-site}.

\begin{figure}[t]
   \centering \
   \psfig{figure=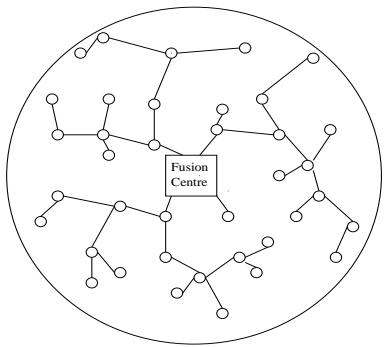,width=5cm,height=4.5cm}
    \caption{An ad hoc wireless sensor network with a fusion center is
     shown. The small circles are the sensor nodes (``motes''), and the
     lines between them indicate wireless links obtained after a
     self-organization procedure.}
   \label{fig:wsn_mesh}
\end{figure}

Event detection is an important task in many sensor network
applications. In general, an event is associated with a change in the
distribution of a related quantity that can be sensed. For example, the
event of a gas leakage at any joints in a pipe causes a change in the
distribution of pressure at the joint and hence can be detected with the
help of pressure sensors. In this paper, we limit our discussion to the
centralized fusion model (see Figure~\ref{fig:wsn_mesh}), in which each
mote, in an event detection network, senses and sends some function of
its observations (e.g., quantized samples) to the fusion center at a
particular rate. The fusion center, by appropriately processing the
sequence of values it receives, makes a decision regarding the state of
nature, i.e., it decides whether a change has occurred or not.

Our problem is that of minimizing the mean detection delay (the delay
between the event occurring and the detection decision at the fusion
center) with a bound on the probability of false alarm. We consider
\emph{a small extent network} in which all the sensors have {\em the
same coverage}, i.e., when the change in distribution occurs it is
observed by all the sensors and the statistics of the observations are
the same at all the sensors. $N$ sensors \emph{synchronously} sample
their environment at a particular rate. Synchronized operation across
sensors is practically possible in networks such as 802.11 WLANs and
Zigbee networks since the access point and the PAN coordinator,
respectively, transmit beacons that provide all nodes with a time
reference. Based on the measurement samples, the nodes send certain
values (e.g., quantized samples) to the fusion center. Each value is
carried by a packet, which is transmitted using a contention--based
multiple access mechanism. Thus, our small extent network problem is a
natural extension of the standard change detection problem
(see~\cite{veeravalli01decentralized-quickest} and the references
therein) to detection over a random access network. The problem of
quickest event detection problem in a large extent network (where the
region of interest is much larger than the sensing coverage of any
sensor) is considered by us in \cite{premkumar-etal09distributed-det}.
Also, a small extent network can be thought of as a cluster in a large
extent network and that the decision maker can be thought of as a
cluster head.

\begin{figure}[t]
   \centering \
   \psfig{figure=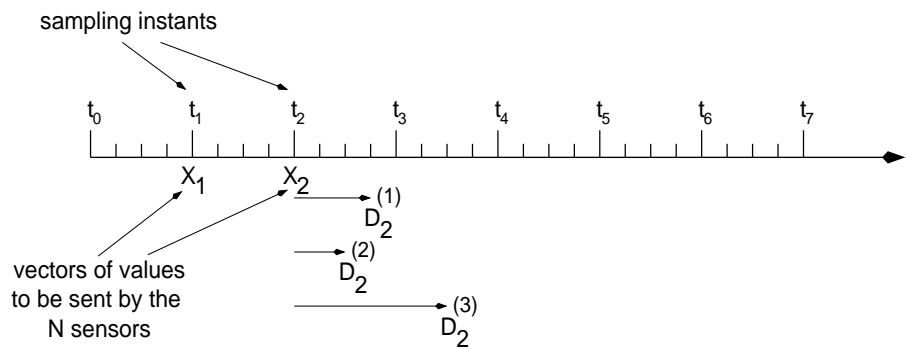,width=12cm,height=4.5cm}
   \caption{The sensors take samples periodically at instants $t_1,
     t_2, \cdots$, and prepare to send to the fusion center a vector
     of values $\mathbf{X}_h = \left[X_h^{(1)}, X_h^{(2)}, \cdots, X_h^{(N)}\right]$ 
     at $t_h$. Each is queued as a packet in
     the queue of the respective node. Due to multiple access delays,
     the packets arrive with random delays at the fusion center; for
     example, for $\mathbf{X}_2$, the delays $D^{(1)}_2, D^{(2)}_2,
     D^{(3)}_2$, for the packets from sensors $1, 2$ and $3$, are
     shown. }
   \label{fig:samples_with_mac_delays}
\end{figure}

In this setting, due to the multiple access network delays between the
sensor nodes and the fusion center, several possibilities arise. In
Figure~\ref{fig:samples_with_mac_delays} we show that although the
sensors take samples synchronously, due to random access delays
the various packets sent by the sensors arrive at the fusion center
asynchronously. As shown in the figure, the packets generated due to
the samples taken at time $t_2$ arrive at the fusion center with a 
delay of $D^{(1)}_2, D^{(2)}_2, D^{(3)}_2$, etc. It can even happen 
that a packet corresponding to the samples taken at time $t_3$ can 
arrive before one of the packets generated due to the samples taken 
at time $t_2$. 

Figure~\ref{fig:system_arch} depicts a general queueing and decision
making architecture in the fusion center. All samples are queued in
per--node queues in a sequencer. The way the sequencer releases the 
packets gives rise to the following three cases, \emph{the first two
of which we study in this paper}.

\begin{figure}[t]
   \centering \
   \psfig{figure=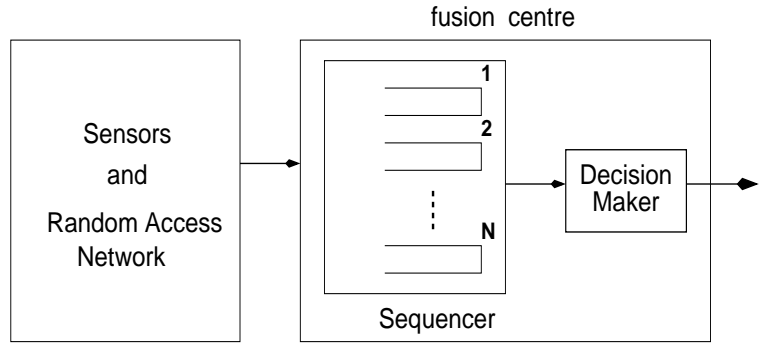,width=10cm,height=4.5cm}
    \caption{A conceptual block diagram of the wireless sensor network
    shown in Figure~\ref{fig:wsn_mesh}. 
    The fusion center has a sequencing 
    buffer which queues out--of--sequence samples and delivers the samples 
    to the decision maker in time--order, as early as possible, batch--wise
    or sample--wise.}
   \label{fig:system_arch}
\end{figure}

\begin{enumerate}
\item The sequencer queues the samples until all the samples of a 
  ``batch'' (a batch is the set of samples generated at a 
  sampling instant) are accumulated; it then releases the entire batch to
  the decision device. The batches arrive to the decision 
  maker in a time sequence order. The decision maker processes
  the batches without knowledge of the state of the network (i.e., 
  reception times at the fusion center, and the numbers of 
  packets in the various queues). We call this, 
  {\em Network Oblivious Decision Making} ({\sf NODM}). In factory 
  and building automation scenarios, there is a major impetus to replace 
  wireline networks between sensor nodes and controllers. In such 
  applications, the first step could be to retain the fusion algorithm 
  in the controller, while replacing the wireline network with a 
  wireless ad hoc network. Indeed, we show that this approach is
  optimal for {\sf NODM}, provided the sampling rate is 
  appropriately optimized. 

\item The sequencer releases samples only in time--sequence order
  but does not wait for an entire batch to accumulate. The decision
  maker processes samples as they arrive. We call this, {\em 
  Network Aware Decision Making} ({\sf NADM}). In {\sf NADM},
  whenever the decision maker receives a sample, it has to roll back 
  its decision statistic to the sampling 
  instant, update the decision statistic with the received sample 
  and then update the decision statistic to the current time slot. 
  The decision maker makes a Bayesian update on the decision statistic 
  even if it does not receive a sample in a slot. Thus, {\sf NADM} 
  requires a modification in the decision making algorithm in the 
  fusion center.

\item The sequencer does not queue any samples. The fusion center 
  acts on the values from the various sampling instants as they arrive, 
  possibly out of order. The formulation of such a problem would be 
  an interesting topic for future research. 
\end{enumerate}

\noindent
\textbf{Our Contributions:} We find that, in the existing literature
on sequential change detection problems (see discussion on related literature 
below), it has been assumed
that, at a sampling instant, the samples from all the sensors reach
the fusion center instantaneously.  As explained above, however, in
our problem the delay in detection is not only due to the detection
procedure requiring a certain amount of samples to make a decision 
(which we call \emph{decision delay}), but also due to the random 
packet delay in the multiple access network (which we call 
\emph{network delay}). We work with a formulation that accounts for 
both these delays, while limiting ourselves to the particular fusion 
center behaviours explained in cases (1) and (2) above.

In Section~\ref{sec:problem_formulation}, we discuss the basic change
detection problem and setup the model. In Section~\ref{sec:no}, we
formulate the change detection problem over a random access network in a
way that naturally includes the network delay. We show that in the case
of {\sf NODM}, the problem objective decouples into a part involving the
network delay and a part involving the optimal decision delay, under the
condition that the sampling instants are periodic. Then, in
Section~\ref{sec:net_delay_model}, we consider the special case of a
network with a star topology, i.e., all nodes are one hop away from the
fusion center and provide a model for contention in the random access
network. In Section~\ref{sec:na}, we formulate the {\sf NADM} problem
where we process the samples as they arrive at the fusion center, but in
a time causal manner. The out--of--sequence packets are queued in a
sequencing buffer and are released to the decision maker as early as
possible. We show in the {\sf NADM} case that the change--detection
problem can be modeled as a Partially Observable Markov Decision Process
(POMDP). We show that a {\em sufficient statistic} for the observations
include the {\em network--state} (which include the queue lengths of the
sequencing buffer, network--delays) and {\em the posterior probability
of change having occurred} given the measurements received and the
network states. As usual, the optimal policy can be characterised via a
Bellman equation, which can then be used to derive insights into the
structure of the policy. We show that the {\em optimal policy is a
threshold on the posterior probability of change and that the threshold,
in general, depends on the network state}. Finally, in
Section~\ref{sec:optimal_parameters} we compare, numerically, the mean
detection delay performance of {\sf NODM} and a simple heuristic
algorithm motivated by {\sf NADM} processing. We show the tradeoff
between the sampling rate $r$ and the mean detection delay. Also, we
show the tradeoff between the number of sensors and the mean detection
delay. 


\noindent
\textbf{Related Literature:} The basic mathematical formulation in
this paper is an extension of the classical problem of sequential
change detection in a Bayesian framework. The centralized version of
this problem was solved by Shiryaev (see \cite{shiryayev}).  The
decentralized version of the problem was introduced by Tenny and
Sandell \cite{tennysandell81detection-distributed}. 
In the decentralized setting, Veeravalli
\cite{veeravalli01decentralized-quickest} provided optimal decision
rules for the sensors and the fusion center, in the context of
conditionally independent sensor observations and a quasi-classical
information structure.  For a large network setting, Niu and Varshney
\cite{stat-sig-proc.niu-varshney05large-wsn} studied a simple
hypothesis testing problem and proposed a \emph{counting rule} based
on the number of alarms. They showed that, for a sufficiently large
number of sensors, the detection performance of the counting rule is
close to that of the optimal rule.  In a recent article on anomaly
detection in wireless sensor networks
\cite{rajasegarar-etal08anamaly-detection}, Rajasegarar et al.\ have
provided a survey of statistical and machine learning based techniques
for detecting various types of anomalies such as sensor faults,
security attacks, and intrusions.
In~\cite{stat-sig-proc.aldosari-moura04decentralized-detection} the
authors consider the problem of decentralized binary hypothesis
testing, where the sensors quantize the observations and the fusion
center makes a binary decision between the two hypotheses.

\vspace{2mm}

{\bf Remark:}
In the existing literature on the topic of optimal sequential 
event detection in wireless sensor networks, to the best of our 
knowledge there has been no prior formulation that incorporates 
multiple access delay between the sensing nodes and the fusion 
center. Interestingly, in this paper we introduce, what can be 
called a \emph{cross layer} formulation involving {\em sequential 
decision theory} and {\em random access network delays.} 
In particular, we encounter the \emph{fork--join queueing model} 
(see, for example, \cite{baccelli-makowski}) that arises in 
distributed computing literature.

\section{The Basic Change Detection Problem}
\label{sec:problem_formulation}
In this section, we introduce the model for the basic change detection 
problem. The notation, we follow, is given here.

 \begin{figure}[t]
   \centering \
   \psfig{figure=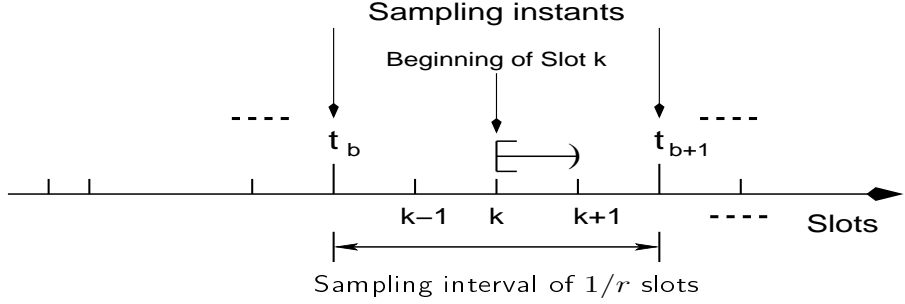,width=12cm,height=4cm}
   \caption{Time evolves over slots. The length of a 
   slot is assumed to be unity. Thus, slot $k$ represents the interval 
   $[k,k+1)$ and the beginning of slot $k$ represents the time instant $k$.
   Samples are taken periodically every $1/r$ slots, starting from $t_1 = 1/r$.
   }
   \label{fig:slot-structure}
 \end{figure}

\begin{itemize}
\item[$\bullet$] Time is slotted and the slots are indexed by $k = 0,1,2\ldots$.
     We assume that the length of a slot is unity and that slot $k$ refers
     to the interval $[k,k+1)$.	Thus, the beginning of slot $k$ indicates the 
     time instant $k$ (see Figure~\ref{fig:slot-structure}).
 \begin{figure}[t]
   \centering \
   \psfig{figure=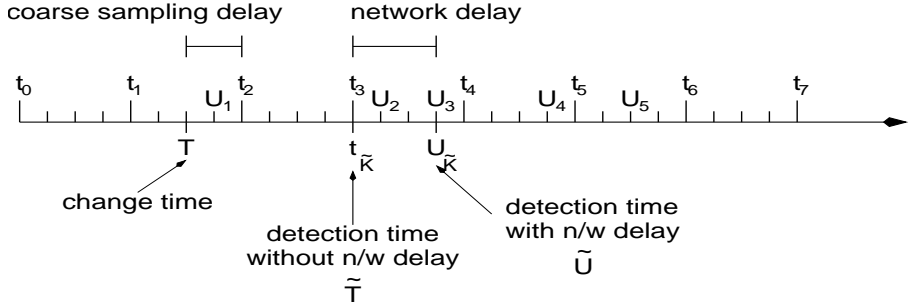,width=12cm,height=4cm}
   \caption{Change time and the detection instants with and without
     network delay are shown. The coarse sampling delay is given by 
     $t_K - T$ where $t_K$ is the first sampling instant after change, 
     and the network delay is given by $U_{\tilde{K}} -
t_{\tilde{K}}$.}
   \label{fig:queueing_decision_instants}
 \end{figure}
\item[$\bullet$] $N$ sensors are \emph{synchronously} sampling at the rate $r$
  samples/slot, i.e., the sensors make an observation every $1/r$
  slots and send their observations to the fusion center. Thus, for
  example, if $r=0.1$, then a sample is taken by a sensor every $10^{th}$ slot. We
  assume that $1/r$ is an integer. The sampling instants are denoted
  $t_1, t_2,\ldots$ (see Figure~\ref{fig:queueing_decision_instants}).
  Define $t_0=0$; note that the first sample is taken at $t_1 = 1/r$.
\item[$\bullet$] The vector of network delays of the batch $b$ is denoted by
  \[\mathbf{D}_b =\left[D^{(1)}_b,D^{(2)}_b, \cdots, D^{(N)}_b\right]\]
  where $D^{(i)}_b \in \{1,2,3,\cdots\}$ is the network delay in slots,
  of the $i$th component of the $b$th batch (sampled at $t_b = b/r$). 
  Also, note that $D_b^{(i)} \geqslant 1$, as it requires one time slot for the 
  transmission of a packet to the fusion center after a successful contention.
\item[$\bullet$] The state of nature $\Theta \in \{0,1\}$. $0$ represents the 
  state {\sf ``before change''} and $1$ represents the state {\sf ``after change''}. 
  It is assumed that 
  the change time $T$ (measured in slots),
  is geometrically distributed i.e.,
  \begin{eqnarray}
  \label{eqn:distribution_of_T}
  \prob{T=0} & = & \rho \nonumber \\
  \text{and, for} \  k\geqslant 1, \ \ \ \prob{T=k\mid T>0} & = & p(1-p)^{(k-1)}.
  \end{eqnarray} 
  The value of $0$ for $T$ accounts for the possibility that the change took 
  place before the observations were made.
\item[$\bullet$] The vector of outputs from the sensor devices at the $b$th
  batch is denoted by
  \[\mathbf{X}_b =\left[X^{(1)}_b,X^{(2)}_b, \cdots, X^{(N)}_b\right]\]
  where $X^{(i)}_b \in \mathcal{X}$ is the $b$th output at the $i$th sensor. 
  Given the state of nature, $X^{(i)}_b$ are assumed to be (conditionally) 
  independent across sensors and i.i.d.\ over sampling instants 
  with probability distributions $F_0(x)$ and $F_1(x)$ before and
  after the change respectively. $\mathbf{X}_1$ corresponds to the
  first sample taken.  In this work, we do not consider the problem of
  optimal processing of the sensor measurements to yield the sensor
  outputs, e.g., optimal quantizers (see \cite{veeravalli01decentralized-quickest}).

\item[$\bullet$] Let ${S}_b$, $b\geqslant 1$, be the state of nature at the $b$th 
  sampling instant and
  $S_0$ the state at time $0$. Then ${S}_b \in \{0,1\}$, with
  \[\prob{S_0=1} = \rho= 1 -\prob{S_0=0}\]
  ${S}_b$ evolves as follows. 
  If ${S}_b=0$ for $b \geqslant 0$, then
  \[S_{b+1} =\left\{\begin{array}{ccc} 1&\mbox{w.p.}&p_r\\0&\mbox{w.p.}&(1-p_r)
  \end{array}\right.\]
  where $p_r = 1-(1-p)^{1/r}$.
  Further, if ${S}_b = 1$, then ${S}_{b+1}=1$. Thus, if ${S}_0 = 0$, then there
  is a change from 0 to 1 at the $K$th sampling instant, where $K$  is geometrically
  distributed. For $b\geqslant 1$,
  \[\prob{K=b} = p_r(1-p_r)^{b-1}\]
\end{itemize}

Each value to be sent to the fusion center by a node is inserted into a packet
which is queued for transmission. It is then transmitted to the fusion center 
by using a contention based multiple access protocol. A node can directly 
transmit its observation to the fusion center or route it through other nodes 
in the system. 
Each packet takes a time slot to transmit.  The MAC protocol and the queues 
evolve over the same time slots. The fusion center makes a decision about
the change depending on whether {\em Network Oblivious} ({\sf NODM}) processing 
or {\em Network Aware} ({\sf NADM}) processing is employed at the fusion center. 
In the case of {\sf NODM} processing, the decision sequence (also called as 
{\em action sequence}), is $A_u, \ u\geqslant 0$, with  
$A_u~\in~\{\mathsf{stop~and~declare~change}(1),$ $\mathsf{take~another~sample}(0)\}$,
where $u$ is a time instant at which a complete batch of $N$ samples corresponding 
to a sampling instant is received by the fusion center. In the case of {\sf NADM} 
processing, the decision sequence is $A_k, \ k\geqslant 0$, with  
$A_k~\in~\{\mathsf{stop~and~declare~change}(1),$ $\mathsf{take~another~sample}(0)\}$,
i.e., a decision about the change is taken at the beginning of every slot.

\section{Network Oblivious Decision Making ({\sf NODM})}
\label{sec:no}
From Figure~\ref{fig:samples_with_mac_delays}, we note that
although all the components of a batch $b$ are generated at 
$t_b = b/r$, they reach the fusion center at times 
$t_b + D_b^{(i)}, i=1,2,\cdots,N$. In {\sf NODM} processing, 
the samples, which are successfully transmitted, are queued 
in a sequencing buffer as they arrive 
(see Figure~\ref{fig:parallel_arch}) and the sequencer releases 
a (complete) batch to the decision maker, as soon as all the 
components of a batch arrive. The decision maker 
makes a decision about the change at the time instants when a 
(complete) batch arrives at the fusion center. In the 
Network Oblivious ({\sf NODM}) processing, the decision maker 
is oblivious to the network and processes the 
batch {\em as though it has just been generated} 
(i.e., as if there is no network, hence the name {\em Network Oblivious Decision Making}). 
We further define (see Figure~\ref{fig:queueing_decision_instants})

 \begin{figure}[t]
   \centering \
   \psfig{figure=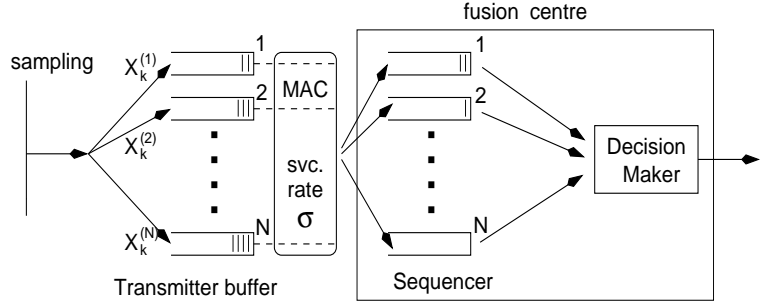,width=10cm,height=4cm}
   \caption{A sensor network model of Figure~\ref{fig:system_arch} with
   one hop communication between the sensor nodes and the fusion center.
   The random access network along with the sequencer is a fork--join
   queueing model.}
   \label{fig:parallel_arch}
 \end{figure}
\vspace{5mm}

\begin{itemize}
\item[$\bullet$]{$U_b, (b \geqslant 1)$:} the random instant at which the fusion
    center receives the complete batch $\mathbf{X}_b$

\item[$\bullet$]{$\widetilde{K} \in \{0,1,\dots\}$:} the batch index at which the
    decision takes place, if there was no network delay. $\widetilde{K} = 0$
    means that the decision 1 ({\sf stop~and~declare~change}) is taken
    before any batch is generated

\item[$\bullet$]{$\widetilde{T}=t_{\widetilde{K}}$:} the random time (a sampling
    instant) at which the fusion center stops and declares change, {\em if
    there was no network delay} 

\item[$\bullet$]{$\widetilde{U}=U_{\widetilde{K}}$:} the random time slot at which the
    fusion center stops and declares change, {\em in the presence of
    network delay}

  \item[$\bullet$]{$D_b=U_b-t_b$:} Sojourn time of the $b$th batch, i.e., the
    time taken for all the samples of the $b$th batch to reach the
    fusion center. Note that $D_b$ is given by $\max\{D_b^{(i)}:i=1,2,\cdots,N\}$.
    Thus, the delay of the batch $\widetilde{K}$ at which the detector
    declares a change is $U_{\widetilde{K}} - t_{\widetilde{K}} = \widetilde{U}-\widetilde{T}$
\end{itemize}
    
We define the following detection metrics. 
\begin{itemize}
\item[{\bf Mean Detection Delay}] defined as the expected number of 
    slots between the change point $T$ and the stopping time instant 
    $\widetilde{U}$ in the presence of {\em coarse sampling} and 
    {\em network} delays,  i.e., {\em Mean Detection Delay} = 
    $\EXP{\left(\widetilde{U}-T\right){\bf 1}_{\{\widetilde{T} \geqslant T\}}}$. 
\item[{\bf Mean Decision Delay}] defined as the expected number of 
    slots between the change point $T$ and the stopping time instant 
    $\widetilde{T}$ in the (presence of {\em coarse sampling} delay 
    and in the) absence of {\em network} delay, i.e., {\em Mean Decision Delay} 
    = $\EXP{\left(\widetilde{T}-T\right){\bf 1}_{\{\widetilde{T} \geqslant T\}}}$. 
\end{itemize}
With the above model and assumptions, we pose the following {\sf NODM} problem:
\emph{Minimize the mean detection delay with a bound on the
  probability of false alarm}, the decision epochs being the time 
instants when a complete batch of $N$ components corresponding to 
a sampling instant is received by the fusion center. In 
Section~\ref{sec:na}, we pose the problem of making a 
decision at every slot based on the samples as they arrive at  
the fusion center.
Motivated by the approach
in~\cite{veeravalli01decentralized-quickest} we use the following
formulation for a given sampling rate $r$
\begin{eqnarray}
  \label{eq:p1}
  & & \min \EXP{(\widetilde{U}-T){\bf 1}_{\{\widetilde{T}\geqslant T\}}} \\ 
\text{such that}  & & \prob{\widetilde{T} < T}\leqslant \alpha \nonumber
\end{eqnarray}
where $\alpha$ is the constraint on the false alarm probability.
 \begin{figure}[h]
   \centering \ 
   \psfig{figure=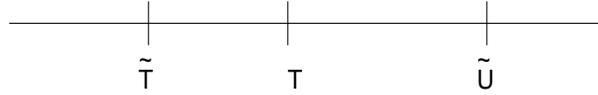,width=80mm,height=12mm}
   \caption{ Illustration of an event of false alarm with
    $\widetilde{T}<T$, but $\widetilde{U}>T$}
   \label{fig:falsealarm}
 \end{figure}
\begin{remarks}
Note that if $\alpha \geqslant 1-\rho$, then the decision making procedure 
can be stopped and an alarm can be raised even before the first observation.
Thus, we assume that $\alpha < 1-\rho$.
\end{remarks}

\begin{remarks}
Note that here we consider $\prob{\widetilde{T}<T}$ as the probability
of false alarm and not $\prob{\widetilde{U}<T}$, i.e., a case as shown
in Figure~\ref{fig:falsealarm} is considered a false alarm. This can be
understood as follows: when the decision unit detects a change at slot
$\widetilde{U}$, the measurements that triggered this inference would be
carrying the ``time stamp'' $\widetilde{T}$, and we infer that the
change actually occurred at or before $\widetilde{T}$. Thus if
$\widetilde{T}<T$, it is an error.   
\end{remarks}

We write the problem defined in Eqn.~\ref{eq:p1} as	
\begin{eqnarray}
  \min_{\Pi_{\alpha}}\EXP{(\widetilde{U}-T){\bf 1}_{\{\widetilde{T}\geqslant T\}}}
\end{eqnarray}
where $\Pi_{\alpha}$ is the set of detection policies for which
$\prob{\widetilde{T}<T} \leqslant \alpha$.

\begin{theorem}
\label{thm:decoupling}
If the sampling is periodic at rate $r$ and the batch sojourn time 
process $D_b$, $b\geqslant 1$, is stationary with mean $d(r)$, then
\begin{eqnarray*}
\min_{\Pi_{\alpha}}\EXP{(\widetilde{U}-T){\bf 1}_{\{\widetilde{T}\geqslant T\}}}
&=& (d(r)+l(r))(1-\alpha) - \rho\cdot l(r) + \frac{1}{r}\min_{\Pi_{\alpha}}\EXP{\widetilde{K}-K}^+
\end{eqnarray*}
where  $l(r)$ is the delay due to (coarse) sampling.
\end{theorem}

\begin{remarks}
  For example in Figure~\ref{fig:queueing_decision_instants}, the
  delay due to coarse sampling is $t_2-T$, $\widetilde{K}-K=3-2=1$, and
  the network delay is $U_3-t_3$.  The stationarity assumption on
  $D_b$, $b\geqslant 1$, is justifiable in a network in which measurements
  are continuously made, but the detection process is started only at
  certain times, as needed. 
\end{remarks}

\emph{Proof:} The following is a sketch of the
  proof (the details are in the Appendix -- I)
\begin{eqnarray*}
  \min_{\Pi_{\alpha}}\EXP{(\widetilde{U}-T){\bf 1}_{\{\widetilde{T}\geqslant T\}}} 
  &=& \min_{\Pi_{\alpha}}\left\{\EXP{(\widetilde{U}-\widetilde{T}) {\bf 1}_{\{\widetilde{T}\geqslant T\}}}+\EXP{\widetilde{T}-T}^+\right\}\\
  &=&\min_{\Pi_\alpha}\left\{ \EXP{D}\left(1-\prob{\widetilde{T}<T}\right)+~\EXP{\widetilde{T}-T}^+\right\}
\end{eqnarray*}
where we have used the fact that under periodic sampling, the queueing
system evolution and the evolution of the statistical decision problem
are independent, i.e., $\widetilde{K}$ is independent of $\{D_1, D_2,
\ldots\}$ and $\EXP{D}$ is the mean stationary queueing delay (of a batch).  By
writing $\EXP{D}=d(r)$ and using the fact that the problem
$\min_{\Pi_{\alpha}}\EXP{\widetilde{T}-T}^+$ is solved by a policy
$\pi^*_\alpha \in \Pi_\alpha$ with $\prob{\widetilde{T}<T}=\alpha$, the
problem becomes
\begin{eqnarray*}
d(r)(1-\alpha) + \min_{\Pi_{\alpha}}\EXP{\widetilde{T}-T}^+
& = & (d(r)+l(r))(1-\alpha) -\rho\cdot l(r)+\frac{1}{r} \min_{\Pi_\alpha} {\EXP{\widetilde{K}-K}^+}
\end{eqnarray*}
where $l(r)$ is the delay due to sampling. Notice that $\min_{\Pi_\alpha} {\EXP{\widetilde{K}-K}^+}$
is the basic change detection problem at the sampling instants.\qed

\begin{remarks}It is important to note that the independence between
  $\widetilde{K}$ and $\{D_1, D_2, \ldots\}$ arises from periodic
  sampling. Actually this is \emph{conditional} independence given the
  rate of the periodic sampling process. If, in general, one considers
  a model in which the sampling is at random times (e.g., the sampling
  process randomly alternates between periodic sampling at two
  different rates or if adaptive sampling is used) then we can view it
  as a time varying sampling rate and the asserted independence will
  not hold. 
\end{remarks}

We conclude that the problem defined in Eqn.~\ref{eq:p1}
effectively decouples into the sum of the 
optimal delay in the original optimal detection problem, i.e.,
$\frac{1}{r}\min_{\Pi_{\alpha}}\EXP{\widetilde{K}-K}^+$ as in
\cite{veeravalli01decentralized-quickest}, a part that captures
the network delay, i.e., $d(r)(1-\alpha)$, and a part that 
captures the sampling delay, i.e., $l(r)(1-\alpha)-\rho l(r)$.

\section{Network Delay Model}
\label{sec:net_delay_model}
 \begin{figure}[t]
   \centering \
   \psfig{figure=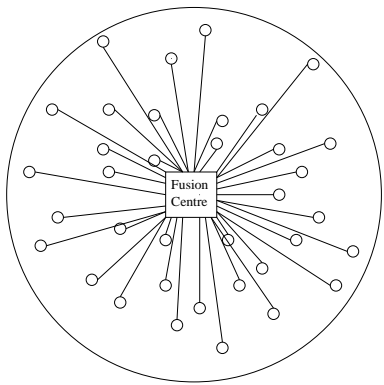,width=5cm,height=5cm}
   \caption{A sensor network with a star topology with the fusion
     center at the hub. The sensor nodes use a random access MAC to
     send their packets to the fusion center.}
   \label{fig:wsn_star}
 \end{figure}
From Theorem~\ref{thm:decoupling}, it is clear that in {\sf NODM} 
processing, the optimal decision device and the queueing system are 
decoupled. Thus, one can employ an optimal sequential change 
detection procedure (see \cite{shiryayev}) for any random
access network (in between the sensor nodes and the fusion center). 
Also, {\sf NODM} is oblivious to the random access network
(in between the sensor nodes and the fusion center) and processes
a received batch as though it has just been generated.  
In the case of {\sf NADM} (which we describe in Section~\ref{sec:na}), 
the decision maker processes samples, keeping network--delays
into account, thus requiring the knowledge of the network dynamics. 
In this section, we provide a simple model
for the random access network, that facilitates the analysis and 
optimisation of {\sf NADM}. 

$N$ sensors form a star topology\footnote{Note that {\em Theorem 1} 
is more general and does not assume a star topology.} (see 
Figure~\ref{fig:wsn_star}) ad hoc wireless sensor network with the
fusion center as the hub. They synchronously sample their environment
at the rate of $r$ samples per slot periodically. At sampling instant
$t_b = b/r$, sensor node $i$ generates a packet containing the sample value
$X^{(i)}_b$ (or some quantized version of it). This packet is then
queued first-in-first-out in the buffer behind the radio link. It is
as if each sample is a \emph{fork} operation that puts a packet into
each sensor queue (see Figure~\ref{fig:parallel_arch}).

The sensor nodes contend for access to the radio channel, and transmit
packets when they succeed. The service is modeled as follows. As long
as there are packets in any of the queues, successes are assumed to
occur at the constant rate of $\sigma~(0<\sigma<1)$ per slot, with the intervals
between the successes being i.i.d., geometrically distributed random
variables, with mean $1/\sigma$. If, at the time a success occurs, there
are $n$ nodes contending (i.e., $n$ queues are nonempty) then the
success is ascribed to any one of the $n$ nodes with equal
probability. 

The $N$ packets corresponding to a sample arrive at random times at
the fusion center. If the fusion center needs to accumulate all the
$N$ components of each sample then it must wait for that component of
every sample that is the last to depart its mote. This is a \emph{join}
operation (see Figure~\ref{fig:parallel_arch}).

\begin{figure}[t]
  \centering
\ \epsfig{file=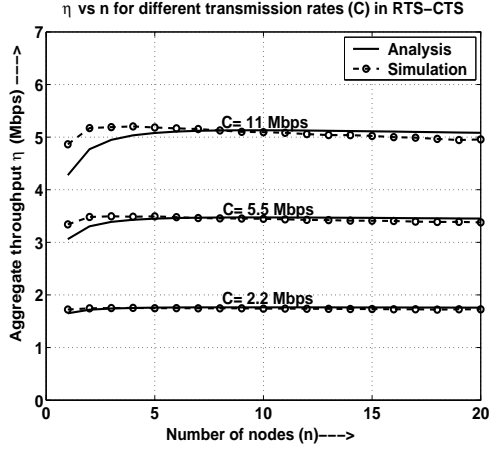,width=65mm,height=60mm}
\caption{The aggregate saturation throughput $\eta$ of an IEEE
  802.11 network plotted against the number of nodes in the network,
  for various physical layer bit rates: 2.2~Mbps, 5.5~Mbps, and
  11~Mbps .  The two curves in each plot correspond to an analysis and
  an NS--2 simulation.}
   \label{fig:sat_thpt_wlan_vs_n}
 \end{figure}

\begin{figure}[t]
\centering
 \begin{minipage}{3.8cm}
   \centering
    \scalebox{.5}[.5]{\includegraphics*[54,201][310,386]{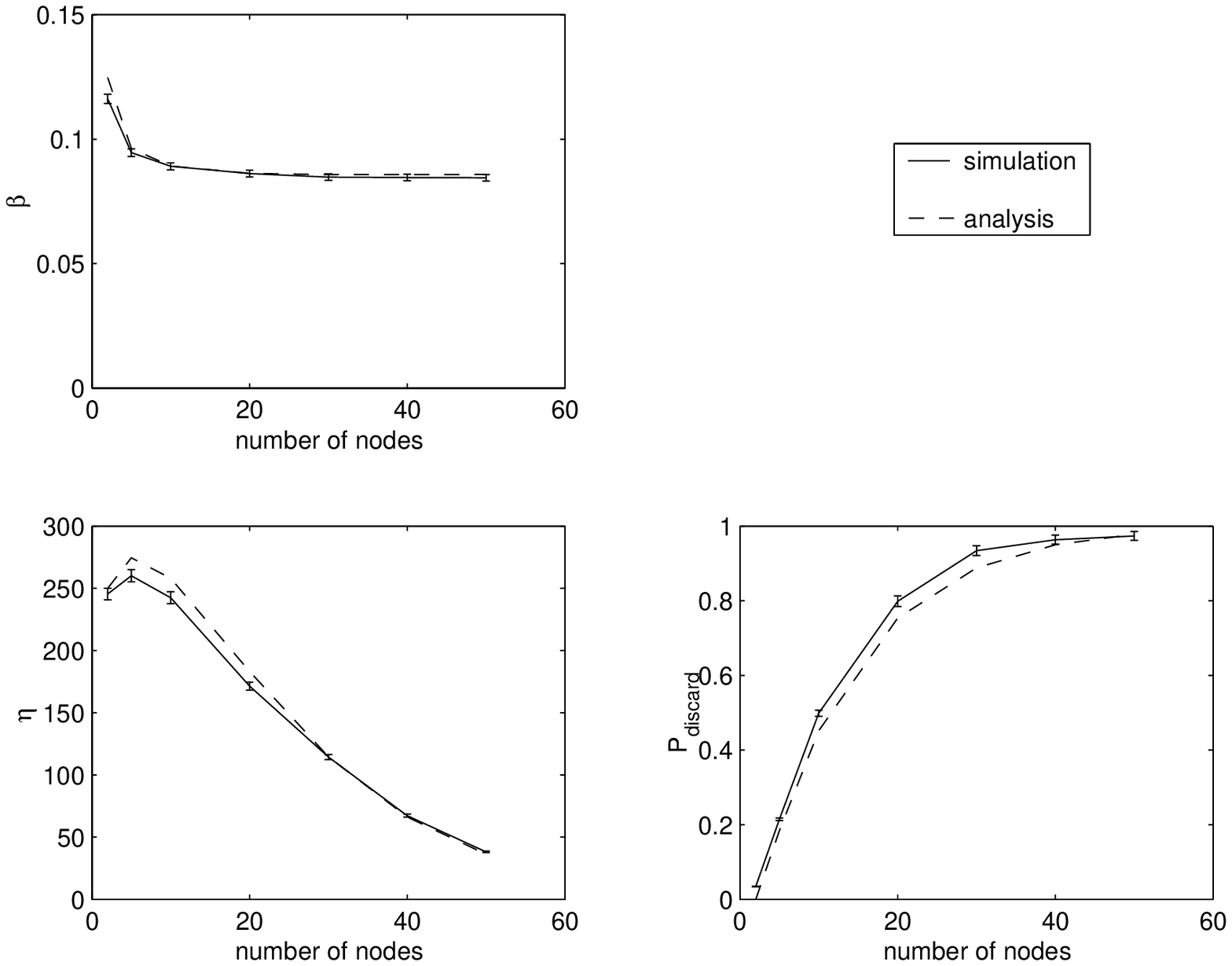}}
   \end{minipage}
\hspace{10mm}
   \begin{minipage}{3.8cm}
     \centering
     \scalebox{.5}[.5]{\includegraphics*[54,201][310,386]{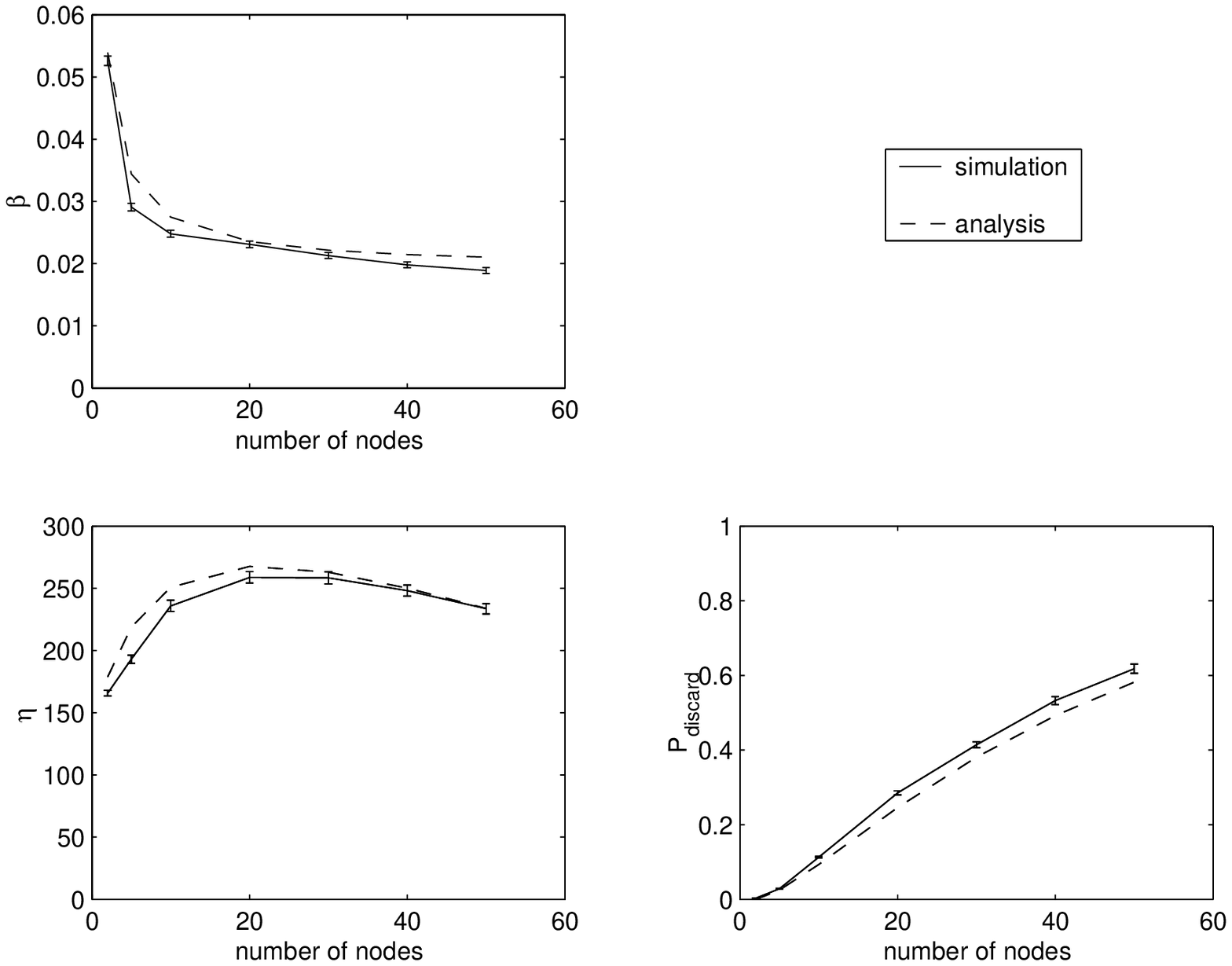}}
   \end{minipage}
   \caption{The aggregate saturation throughput $\eta$ of an IEEE
     802.15.4 star topology network plotted against the number of nodes
     in the network. Throughput obtained with default backoff
     parameters is shown on the left and that obtained with backoff
     multiplier $=3$, is shown on the right. The two curves in each
     plot correspond to an analysis and an NS--2 simulation.}
   \label{fig:thruput_zigbee}
 \end{figure}
It is easily recognized that our service model, in the case of 
{\sf NODM} is the discrete time equivalent of generalized processor 
sharing (GPS -- see, for example, \cite{books.kmk04analytical}),
which can be called the FJQ-GPS (fork-join queue (see \cite{baccelli-makowski}) 
with GPS service). In the case of {\sf NADM}, the service model 
is just the GPS.  

In IEEE~802.11
networks and IEEE~802.15.4 networks, if appropriate parameters are
used, then the adaptive backoff mechanism can achieve a throughput
that is roughly constant over a wide range of $n$, the number of
contending nodes.  This is well known for the CSMA/CA implementation
in IEEE~802.11 wireless LANs; see, for example,
Figure~\ref{fig:sat_thpt_wlan_vs_n} \cite{books.kmk08wireless}. For each physical layer rate, the
network service rate remains fairly constant with increasing number of
nodes.  From Figure~\ref{fig:thruput_zigbee} (taken
from~\cite{chandramani-thesis}) it can be seen that with the default
backoff parameters, the saturation throughput of a star topology
IEEE~802.15.4 network decreases with the number of nodes $N$, but with
the backoff multiplier $=3$, the throughput remains almost constant
from $N=10$ to $N=50$~\cite{chandramani-thesis}; thus, in the latter
case our GPS model can be applicable.

\begin{theorem}
\label{thm:fjq-gps_stationary_delay}
The stationary delay $D$ is a proper random variable with finite
mean if and only if $N r < \sigma$.
 \end{theorem}
\emph{Proof:} See Appendix -- II.\qed

\vspace{1em}
\noindent
Thus, for the FJQ--GPS queueing system to be stable, 
the sampling rate $r$ is chosen such that $r < \frac{\sigma}{N}$. 

\section{Network Aware Decision Making ({\sf NADM})}
\label{sec:na}
In Section~\ref{sec:no}, we formulated the problem of {\sf NODM}
quickest change detection over a random access network, and showed that
(when the decision instants are $U_k$, as shown in
Figure~\ref{fig:queueing_decision_instants}) the optimal decision maker
is independent of the random access network, under periodic sampling.
Hence, the Shiryaev procedure, which is shown to be delay optimal in the
classical change--detection problem (see \cite{shiryayev}), can be
employed in the decision device independently of the random access
network. It is to be noted that the decision maker in the {\sf NODM}
case, waits for a complete batch of $N$ samples to arrive, to make a
decision about the change. Thus, the mean detection delay of the {\sf
NODM} has a network--delay component corresponding to a batch of $N$
samples.  In this section, we provide an alternative mechanism of fusion
at the decision device called {\em Network Aware Decision Making} ({\sf
NADM}), in which the fusion algorithm does not wait for an entire batch
to arrive, and processes the samples as soon as they arrive, but in a
time--causal manner.  
 
We now describe the processing in {\sf NADM}. Whenever a node (successfully)
transmits a sample across the random access network, it is delivered to 
the decision maker if the decision maker has received all the samples 
from all the batches generated earlier. Otherwise, the sample is an 
out--of--sequence sample, and is queued in the sequencing buffer. It 
follows that, whenever the (successfully) transmitted sample is the 
last component of the batch that the decision maker is looking for, 
then the {\em head of line} (HOL) components, if any, in the queues 
of the sequencing buffer are also delivered to the decision maker. 
This is because, these HOL samples belong to the \emph{next} batch 
that the decision maker should process. The decision maker makes a 
decision about the change {\em at the beginning of every time slot}, 
irrespective of whether it receives a sample or not. In {\sf NADM}, 
whenever the decision maker receives a sample, it takes into account 
the network--delay of the sample while computing the decision statistic. 
The network--delay is a part of the state of the queueing system 
which is available to the decision maker. Thus, unlike NODM, {\em the state of the 
queueing system also plays a role in decision making}.

In Section~\ref{subsec:notation}, we define the state of the queueing system. 
In Section~\ref{subsec:evol-Q}, we define the dynamical system whose change 
of state (from $0$ to $1$) is the subject of interest to us. We define the 
state of the dynamical system as a tuple that contains 
the {\em queueing state}, the {\em state of nature}, and {\em a delayed state
of nature}. The delayed state of nature is included in the state of the 
system so that the (delayed) sensor--observations that the decision maker 
receives at time instant $k+1$ depend only on the state, the control, and the 
noise of the system at time instant $k$,
a property which is desirable to define a sufficient statistic
(see page 244, \cite{books.bertsekas00a}).
We explain the evolution of the state
of the dynamical system in Section~\ref{subsec:system-state-evolution-model}.
In Section~\ref{subsec-the-NADM-change-detection-problem}, we formulate
the {\sf NADM} change detection problem and we find a sufficient statistic for the 
observations in Section~\ref{subsec:sufficient-statistic}. In 
Section~\ref{subsec:optimal-stopping-time}, we provide the optimal 
decision rule for the {\sf NADM} change detection problem.
\subsection{Notation and State of the Queueing System}
\label{subsec:notation}
Recall the notation introduced in Section~\ref{sec:problem_formulation}.
Time progresses in slots, indexed by $k=0,1,2\cdots$; the beginning 
of slot $k$ is the time instant $k$. Also, the time instant 
just after the beginning of time slot is denoted by $k+$ 
\footnote{Note that the notation $t+$ denotes a time embedded to the
right of $t$ and is different from the notation $(x)^+$.
Recall that $(x)^+ := \max\{x,0\}$.}. 
Recall that the nodes take samples at the instants 
$1/r$, $2/r$, $3/r$, $\cdots$. We define the state of the queueing system here. Note that the queueing system evolves over slots.
\begin{itemize}
\item[$\bullet$] $\lambda_k \in \{1,2,\cdots, 1/r\}$ denotes the number 
      of time slots to go for the next sampling instant, at the beginning
      of time slot $k$ (see Figure~\ref{fig:lambda-delta}). Thus,

      \begin{eqnarray}
      \label{eqn:lambda-defn}
      \lambda_k & := & \frac{1}{r} - \left(k\mod\frac{1}{r}\right). 
      \end{eqnarray}
      Thus, $\lambda_0 = \frac{1}{r}, \lambda_1 = \frac{1}{r}-1,
      \cdots$, and at the sampling instants $t_b$, $\lambda_{t_b} = \frac{1}{r}$.

 \begin{figure}[t]
   \centering \
   \psfig{figure=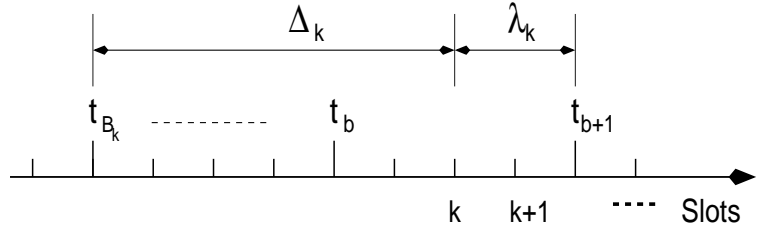,width=10cm,height=3cm}
   \caption{At time $k$, the decision maker expects samples (or processes
   samples) from batch $B_k$. Also, at time $k$, $\lambda_k$ is the 
   number of slots to go for the next sampling instant and
   $\Delta_k$ is the number of slots back at which batch $B_k$ is generated.}
\label{fig:lambda-delta}
 \end{figure}

\item[$\bullet$] $B_k \in \{1,2,3,\cdots\}$ denotes the index of the 
      batch that is expected to be or is being processed by the decision maker 
      at the beginning of time slot $k$. Note $B_0 = B_1 = \cdots = 
      B_{1/r} = 1$. Also, note that the batch $B_k$ is generated at 
      time instant $B_k/r$. 

 \begin{figure}[t]
   \centering \
   \psfig{figure=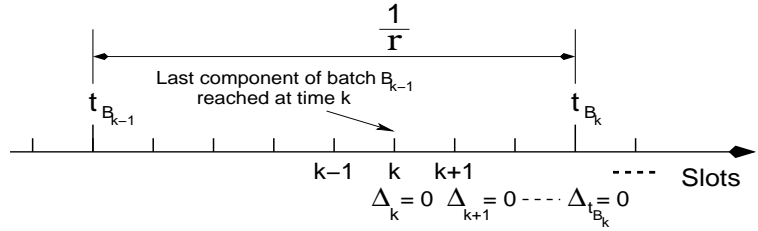,width=10cm,height=3cm}
   \caption{Illustration of a scenario in which $\Delta_k = 0$. If the 
    last component from batch $B_{k-1}$ is received at $k$, and if there 
    is no sampling instant between $t_{B_{k-1}}$ and $k$, then $\Delta_k = 0$.
    Also, note in this case that $\Delta_{k} = \Delta_{k+1} = \cdots = \Delta_{t_{B_k}} = 0$.
    In this scenario, at time instants $k,k+1,\cdots,t_{B_k}$, all the queues 
    at the sensor nodes and at the sequencer are empty, and at time instant
    $t_{B_k}+$, all sensor node queues have one packet which is generated at $t_{B_k}$.  
   }
 \label{fig:delta-future-b}
 \end{figure}
\item[$\bullet$] $\Delta_k \in \{0,1,2,\cdots\}$ denotes the delay 
      in number of time slots between the time instants $k$ and 
      $B_k/r$ (see Figure~\ref{fig:lambda-delta}). 
      \begin{eqnarray}
      \label{eqn:delta-defn}
      \Delta_k & := & \max\left\{k-\frac{B_k}{r}, 0\right\}.
      \end{eqnarray}
      Note that the batches of samples taken after $B_k/r$ and up to 
      (including) $k$ are queued either in the sensor node queues or 
      in the sequencing buffer in the fusion center. If at time $k$, 
      the fusion center receives a sample which 
      is the last sample from batch $B_{k-1}$, then $B_{k} = B_{k-1}+1$.
      If the sampling instant of the $B_k$th batch is later than $k$ 
      (i.e., $B_k/r > k$), then $\Delta_k = 0$ (up to time 
      $B_k/r$ at which instant, a new batch is generated). This corresponds 
      to the case, when all the samples generated up to time slot $k$, 
      have already been processed by the decision maker 
      (see Figure~\ref{fig:delta-future-b}). In particular, 
      $\Delta_0 = \Delta_1 = \cdots = \Delta_{\frac{1}{r}-1} = 0$.

\begin{figure}[t]
   \centering \
   \psfig{figure=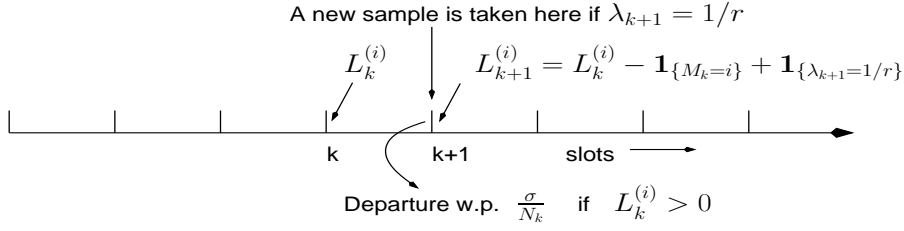,width=12cm,height=3cm}
   \caption{The evolution of $L_k^{(i)}$ from time slot $k$ to time slot $k+1$. 
      If during time slot $k$, node $i$ transmits (successfully) a packet to the 
      fusion center (i.e., $M_k = i$), then that packet is flushed out of its 
      queue at the end of time slot $k$. Also, a new sample is generated 
      (every $1/r$ slots) exactly at the beginning of a time slot. Thus,
      $L^{(i)}_{k+1}$, the queue length of sensor node $i$ just after the 
      beginning of time slot $k+1$ (i.e., at $(k+1)+$) is given by 
      $L_{k+1}^{(i)} = L_k^{(i)} - {\bf 1}_{\{M_k=i\}} + {\bf 1}_{\{\lambda_{k+1} = 1/r\}}$.} 
   \label{fig:na_processing_embedding_sensor_side}
\end{figure}

\item[$\bullet$] $L_k^{(i)} \in \{0,1,2,\cdots\}$ denotes the queue 
     length of the $i$th sensor node just after the beginning of time 
     slot $k$ (i.e., at time instant $k+$). The vector of queue lengths is 
     ${\bf L}_k = [L_k^{(1)}, L_k^{(2)}, \cdots, L_k^{(N)}]$.
     Let $N_k \in \{0,1,2,\cdots,N\}$ be the number of non--empty
     queues at the sensor nodes, just after the beginning of time slot $k$.
     \[N_k := \sum_{i=1}^N {\bf 1}_{\{L_k^{(i)}>0\}}\] 
     i.e., the number of sensor nodes that contend for slot $k$ is $N_k$.
     Hence, using the network model we have provided in Section~\ref{sec:net_delay_model},
     the evolution of $L_k^{(i)}$ (see Figure~\ref{fig:na_processing_embedding_sensor_side})
     is given by the following: 
\setlength{\extrarowheight}{0.3cm}
     \begin{eqnarray*}
     L_0^{(i)} & = & 0 \\ 
     L_{k+1}^{(i)} & = & \left\{
                         \begin{array}{lll} 
                         L_{k}^{(i)} + {\bf 1}_{\{\lambda_{k+1}=1/r\}} & \ \ \mathrm{w.p.} \ 1 & \text{if} \ N_k = 0, \\  
                         L_{k}^{(i)} + {\bf 1}_{\{\lambda_{k+1}=1/r\}} & \ \ \mathrm{w.p.} \ (1-\sigma) & \text{if} \ N_k > 0,\\  
                         \max\{L_{k}^{(i)}-1,0\} + {\bf 1}_{\{\lambda_{k+1}=1/r\}} & \ \ \mathrm{w.p.} \ \frac{\sigma}{N_{k}} & \text{if} \ N_k > 0.  
                         \end{array}
     \right. 
     \end{eqnarray*}

\setlength{\extrarowheight}{-0.3cm}
     Note that when all the samples generated up to time slot $k$ have already been processed by the 
     decision maker and $k$ is not a sampling instant, i.e., $\Delta_k = 0$ and 
     $\lambda_{k} \neq 1/r$, then ${\bf L}_k = {\bf 0}$ (as there are no outstanding 
     samples in the system). For e.g., ${\bf L}_1 = {\bf L}_2 = \cdots = 
     {\bf L}_{1/r-1} = {\bf 0}$. Also, note that just after sampling instant $t_b$, $L_{t_b}^{(i)} \geqslant 1$.

\item[$\bullet$] $M_k \in \{0,1,2,\cdots,N\}$ denotes the index of the node
     that successfully transmits in slot $k$. $M_k=0$ means that there
     is no successful transmission in slot $k$. Thus,
     from the network model we have provided in Section~\ref{sec:net_delay_model},
     we note that
     \begin{eqnarray*}
     M_{k} & = & \left\{
                         \begin{array}{lll} 
                         0 & \mathrm{w.p.} \ 1                  & \text{if} \ N_k = 0\\  
                         0 & \mathrm{w.p.} \ (1-\sigma)         & \text{if} \ N_k > 0\\  
                         j & \mathrm{w.p.} \ \frac{\sigma}{N_k} & \text{if} \ L_k^{(j)} > 0, \ j=1,2,\cdots,N
                         \end{array}
     \right. 
     \end{eqnarray*}
      
\begin{figure}[t]
   \centering \
   \psfig{figure=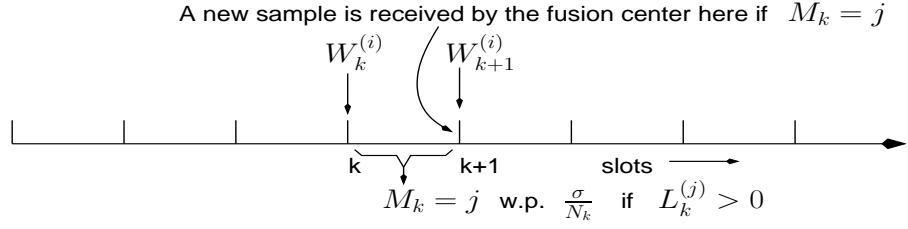,width=12cm,height=3cm}
   \caption{The evolution of $W_k^{(i)}$ from time slot $k$ to time slot $k+1$. 
            If a sample from node $i$ is transmitted (successfully) during time 
            slot $k$ (i.e., $M_k = i$), then it is received by the fusion center 
            at the end of time slot $k$ (i.e. at $(k+1)-$). If this sample is from
            batch $B_k$, it is passed on to the decision maker. Otherwise, it is 
            queued in the sequencing buffer, in which case $W_{k+1}^{(i)} = W_k^{(i)} + 1$. 
            On the other hand, if a sample from some other node $j$ is transmitted 
            (successfully) during time slot $k$ (i.e., $M_k = j\neq i$), and if this 
            sample is the last component to be received from batch $B_k$ by the fusion 
            center, then the HOL packet of the $i$th sequencing queue, if any, is also 
            delivered to the decision maker. Thus, in this case, 
            $W_{k+1}^{(i)} = \max\{W_k^{(i)} - 1,0\}$. Note that $W_{k+1}^{(i)}$ refers 
            to the queue length corresponding to node $i$ at the sequencer, at the 
            beginning of time slot $k+1$.}
   \label{fig:na_processing_embedding_fc_side}
\end{figure}

\item[$\bullet$] $W_k^{(i)} \in \{0,1,2,\cdots\}$ denotes the queue 
     length of the $i$th sequencing buffer at time $k$. The vector of 
     queue lengths is given by ${\bf W}_k = [W_k^{(1)}, W_k^{(2)}, 
     \cdots, W_k^{(N)}]$. Note that ${\bf W}_k = {\bf 0}$ if $\Delta_k 
     = 0$, i.e., the sequencing buffer is empty if there are no 
     outstanding samples in the system. In particular, ${\bf W}_0 = 
     {\bf W}_1 = \cdots = {\bf W}_{\frac{1}{r}} = {\bf 0}$. The 
     evolution of $W_k^{(i)}$ is explained in 
     Figure~\ref{fig:na_processing_embedding_fc_side}. If a sample from
     node $i$ of a batch later than $B_k$ is successfully transmitted
     during slot $k$, then $W_{k+1}^{(i)} = W_k^{(i)}+1$. If a sample from
     node $j$ of batch $B_k$ is successfully transmitted and if it is 
     the last sample to be received from batch $B_k$, then the queue
     lengths of sequencing buffer are decremented by 1, i.e.,
     $W_{k+1}^{(i)} = \max\{W_k^{(i)}-1,0\}$.   
     
\item[$\bullet$] $R_k^{(i)} \in \{0,1\}$ denotes whether the sample 
     $X_{B_k}^{(i)}$ has been received and processed by the decision 
     maker at time $k$. $R_k^{(i)} = 0$ means that the sample 
     $X_{B_k}^{(i)}$ has not yet been received by the decision maker 
     and $R_k^{(i)} = 1$ means that the sample $X_{B_k}^{(i)}$ has been 
     received and processed by the decision maker. The vector of 
     $R_k^{(i)}$s is given by ${\bf R}_k = [R_k^{(1)}, R_k^{(2)}, 
     \cdots, R_k^{(N)}]$. Note that, if $R_k^{(i)} = 0$, ${W}_k^{(i)} 
     = 0$, i.e., the $i$th sequencing buffer is empty if the sample 
     expected by the decision maker has not yet been transmitted. Also 
     note that when $\Delta_k = 0$, ${\bf R}_k = {\bf 0}$, as the 
     samples from the current batch $B_k$ have yet to be generated or 
     have just been generated.        
\end{itemize}

We now relate the queue lengths $L_k^{(i)}$ and $W_k^{(i)}$. Note that 
at the beginning of time slot $k$, $\left\lfloor\frac{k}{1/r}\right\rfloor$ 
batches have been generated so far, of which $B_k - 1$ batches are 
completely received by the decision maker. In batch $B_k$, $i$th sample
is received by the decision maker if $R_k^{(i)} = 1$. Hence, at time
$k$, $B_k - 1 + R_k^{(i)}$ samples generated by node $i$ 
have been processed by the decision maker and the remaining samples are in the 
sensor and sequencing queues. Thus, we have 
\begin{eqnarray}
     \label{eqn:L-W}
     L_k^{(i)} + W_k^{(i)} & = & \left\lfloor\frac{k}{1/r}\right\rfloor - \left(B_k - 1\right) - R_k^{(i)} \\
      & = & \left\lfloor\frac{k - B_k/r + 1/r}{1/r}\right\rfloor   - R_k^{(i)} \nonumber \\ 
      & = & \left\{
            \begin{array}{ll}
\left\lfloor\frac{\Delta_k}{1/r}\right\rfloor  + 1 - R_k^{(i)} & \hspace{11mm}\text{if} \ k > B_k/r \\
1 - R_k^{(i)} & \hspace{11mm}\text{if} \ k = B_k/r \\
- R_k^{(i)} & \hspace{11mm}\text{if} \ k < B_k/r 
            \end{array}
\right. \nonumber 
     \end{eqnarray}
Recalling the definition of $\Delta_k$, we write the above Eqn. as, 
     \begin{eqnarray}
     \label{eqn:L-W-D}
     L_k^{(i)} + W_k^{(i)} 
      & = & \left\{
            \begin{array}{ll}
\left\lfloor\frac{\Delta_k}{1/r}\right\rfloor  + 1 - R_k^{(i)} & \hspace{11mm}\text{if} \ \Delta_k > 0 \\
1  & \hspace{11mm}\text{if} \ \Delta_k = 0, \lambda_k = 1/r \\
0 & \hspace{11mm}\text{if} \ \Delta_k = 0, \lambda_k \neq 1/r. 
            \end{array}
\right.
     \end{eqnarray}
Note that in the above Eqn. $\Delta_k=0,\lambda_k=1/r$ (or equivalently $k=B_k/r$), 
corresponds to the case when the samples of batch $B_k$ have just been taken 
and all the samples from all previous batches have been processed. Thus,
in this case $L_k^{(i)} = 1$ (as $W_k^{(i)} = 0$). In the case of 
$\Delta_k=0,\lambda_k\neq1/r$ (or equivalently $k<B_k/r$), all the samples 
from all previous batches have been processed and a new sample from batch 
$B_k$ is not taken yet. Thus, in this case $L_k^{(i)} = 0$ (and $W_k^{(i)}=0$).   
Hence, given ${\bf Q}_k = \left[\lambda_k, B_k, \Delta_k, {\bf W}_k, {\bf R}_k\right]$,
the queue lengths $L_k^{(i)}$s can be computed as 
     \begin{eqnarray}
     \label{eqn:phi-L-m}
     L_k^{(i)} \ = \  \phi_{L^{(i)}}({\bf Q}_k) & := &  
       \left\{
            \begin{array}{ll}
\left\lfloor\frac{\Delta_k}{1/r}\right\rfloor  + 1 - R_k^{(i)} - W_k^{(i)} 
   & \hspace{2mm}\text{if} \ \Delta_k > 0 \\
1  & \hspace{2mm}\text{if} \ \Delta_k = 0, \lambda_k = 1/r \\
0  & \hspace{2mm}\text{if} \ \Delta_k = 0, \lambda_k \neq 1/r 
            \end{array}
\right. .\\
     \label{eqn:phi-n} 
   \text{Also}, \ N_k \ = \ \phi_N({\bf Q}_k) & := & \sum_{i=1}^N {\bf 1}_{\{ \phi_{L^{(i)}}({\bf Q}_k)  > 0\}}. 
     \end{eqnarray}
Thus, the state of the queueing system at time $k$, can be expressed as 
${\bf Q}_k = \left[\lambda_k, B_k, \Delta_k, {\bf W}_k, {\bf R}_k\right]$.
Note that the {\em decision maker can observe the state ${\bf Q}_k$ 
perfectly}. The evolution of the queueing system is explained in the next subsection.

\subsection{Evolution of the Queueing System}
\label{subsec:evol-Q}
The evolution of the queueing system from time $k$ to time $k+1$ depends
only on $M_k$, the success/no--success of contention on the random
access channel. Note that the evolution of $\lambda_k$ is deterministic
and that of $\Delta_k$ depends on $B_k$. Hence, to describe the
evolution of ${\bf Q}_k$, it is enough to explain the evolution of $B_k,
{\bf W}_k$, and ${\bf R}_k$ for various cases of $M_k$. Let ${\bf
Y}_{k+1} \in \{\emptyset\}\cup\left(\cup_{n=1}^N\mathcal{X}^n\right)$
denote the vector of samples received, if any, by the decision maker at
the beginning of slot $k+1$ (i.e., the decision maker can receive a
vector of $n$ samples where $n$ ranges from 0 to $N$).
 
\vspace{4mm}

At the beginning
of time slot $k+1$, the following possibilities arise:
\begin{itemize}
\item[$\bullet$] {\bf No successful transmission:} This corresponds to 
      the case i) when all the queues are empty at the sensor nodes 
      ($N_k = 0$), or ii) when some queues are non--empty at the sensor 
      nodes ($N_k > 0$), either no queue attempts, or there is more than 
      one attempt (resulting in a collision). In either case, $M_k = 0$ 
      and the decision maker does not receive any sample, i.e., 
      ${\bf Y}_{k+1} = \emptyset$. In this case, it is clear that 
      $B_{k+1} = B_k$, ${\bf W}_{k+1} = {\bf W}_k$, and ${\bf R}_{k+1} = 
      {\bf R}_k$.

\item[$\bullet$] {\bf Successful transmission of node $j$'s sample from a later batch:} 
      This corresponds to the case, when the decision maker has already received the 
      $j$th component of the current batch $B_k$ (i.e., $R_k^{(j)} = 1$) and that it has 
      not received some sample, say $i \neq j$, from the batch $B_k$ 
      (i.e., $R_k^{(i)} = 0$, for some $i$). The received sample (is an out--of--sequence
      sample and) is queued in the 
      sequencing buffer ($W_{k+1}^{(j)} = W_k^{(j)} + 1$). Thus, in this case, 
      $M_k = j$ and the decision maker does not receive 
      any sample, i.e., ${\bf Y}_{k+1} = \emptyset$. 
      In this case, it is clear that $B_{k+1} = B_k$, ${\bf W}_{k+1} =
      {\bf W}_k + {\bf e}_j$, and ${\bf R}_{k+1} = {\bf R}_k$.
 
\item[$\bullet$] {\bf Successful transmission of node $j$'s current sample which 
      is not the last component of the batch $B_k$:} This corresponds to 
      the case when the decision maker has not received the $j$th component of the 
      batch $B_k$ before time slot $k$ ($R_k^{(j)} = 0$), and that it has received all the 
      samples that are generated earlier than that of the successful sample. Also,
      the fusion center is yet to receive some other component of batch $B_k$ (i.e.,
      $\sum_{i=1}^N R_k^{(i)} < N-1$). 
      Thus, in this case, $M_k = j$ and the decision maker receives the 
      sample ${\bf Y}_{k+1} = X_{B_k}^{(j)}$. 
      In this case, it is clear that $B_{k+1} = B_k$, ${\bf W}_{k+1} =
        {\bf W}_k$, and ${\bf R}_{k+1} = {\bf R}_k+{\bf e}_j$.  

\item[$\bullet$] {\bf Successful transmission of node $j$'s current sample which 
      is the last component of the batch $B_k$:} This corresponds to the case when 
      the decision maker has not received the $j$th component of the batch $B_k$ 
      before time slot $k$ ($R_k^{(j)} = 0$), and that it has received all the samples that are 
      generated earlier than that of the successful sample. Also, this sample is 
      the last component of batch $B_k$, that is received by the fusion center.  
      (i.e., $\sum_{i=1}^N R_k^{(i)} = N-1$). In this case (along with the 
      received sample), the queues of the sequencing buffer deliver the 
      {\em head of line} (HOL) components (which correspond to the batch index $B_k+1$), 
      if any, to the decision maker and the queues are decremented by one 
      ($W_{k+1}^{(i)} = \max\{W_k^{(i)} - 1,0\}$).   
      Thus, $M_k = j$ and the decision maker receives the 
      vector of samples ${\bf Y}_{k+1} = 
      \left[X_{B_k}^{(j)},X_{B_k+1}^{(i'_1)}, 
      X_{B_k+1}^{(i'_2)}, \cdots, X_{B_k+1}^{(i'_{n-1})}\right]$ where 
      $W_k^{(i)} > 0$ for $i \in \{i'_1,i'_2,\cdots i'_{n-1}\}$, 
      and $W_k^{(i)} = 0$ for $i \notin \{i'_1,i'_2,\cdots i'_{n-1}\}$. 
      In this case, $B_{k+1} = B_k+1$, ${\bf W}_{k+1} =
        {\bf W}_k-{\bf e}_{i_1'}-{\bf e}_{i_2'}-\cdots-{\bf e}_{i_{n-1}'}$,
          and ${\bf R}_{k+1} = {\bf e}_{i_1'}+{\bf
e}_{i_2'}+\cdots+{\bf e}_{i_{n-1}'}$.  
\end{itemize}
Thus, the state of the queueing system at time $k+1$ can be described by 
\begin{eqnarray*}
{\bf Q}_{k+1} & = & \phi_{\bf Q}({\bf Q}_k,M_k)\nn 
             & := & 
 \left[\phi_{\lambda}({\bf Q}_k,M_k),
 \phi_{B}({\bf Q}_k,M_k),
 \phi_{\Delta}({\bf Q}_k,M_k),
 \phi_{\bf W}({\bf Q}_k,M_k),
 \phi_{\bf R}({\bf Q}_k,M_k)\right].
\end{eqnarray*}

In the next subsection, we provide a model of the dynamical system whose
state has the state of nature $\Theta_k$ as one of its constituents. The
quickest detection of change of $\Theta_k$ from 0 to 1 (at a random
time $T$) is the focus of this paper.
 
\subsection{System State Evolution Model} 
\label{subsec:system-state-evolution-model} 
Let $\Theta_k \in \{0,1\}, k \geqslant 0$, be the state of nature at the
beginning of time slot $k$. Recall that $T$ is the change point, i.e.,
for $k<T$, $\Theta_k=0$ and for $k\geqslant T$, $\Theta_k = 1$, and that
the distribution of $T$ is given in Eqn.~\ref{eqn:distribution_of_T}.
The state $\Theta_k$ is observed only through the sensor measurements,
but these are delayed. We will formulate the optimal {\sf NADM} change  
detection problem as a partially observable Markov decision process
(POMDP) with the delayed observations.
The approach and
the terminology used here is in accordance with the stochastic control
framework in \cite{books.bertsekas00a}. 
At time $k$, a sample, if any,
that the decision maker receives is generated at time $B_k/r < k$ (i.e., samples arrive at
the decision maker with a network--delay of $\Delta_k = k-\frac{B_k}{r}$
slots). 
To make an inference about $\Theta_k$ from the
sensor measurements, we must consider the vector of states of nature 
that corresponds to the time instants $k-\Delta_k, k-\Delta_k+1, \cdots,
k$. We define the vector of states at time $k$, ${\bf\Theta}_k := [\Theta_{k-\Delta_k},
\Theta_{k-\Delta_k+1}, \cdots, \Theta_k]$. Note that the length of the
vector depends on the network--delay $\Delta_k$. When $\Delta_k > 0$,
${\bf\Theta}_k =  [\Theta_{\frac{B_k}{r}},\Theta_{\frac{B_k}{r}+1},
\cdots, \Theta_k]$,
and when $\Delta_k = 0$, ${\bf\Theta}_k$ is just [$\Theta_k]$.

Consider the discrete--time
system, which at the beginning of time slot $k$ is described by the
state  
\begin{eqnarray*} 
\Gamma_k & = & [{\bf Q}_k, {\bf \Theta}_k],
\end{eqnarray*} 
where we recall that
\begin{eqnarray*} 
{\bf Q}_{k} & = & \bigg[\lambda_k, B_k, \Delta_k, {\bf W}_k, {\bf R}_k\bigg],\\
{\bf\Theta}_k & = & [\Theta_{k-\Delta_k}, \Theta_{k-\Delta_k+1},
\cdots, \Theta_k].
\end{eqnarray*}
Note that $\Gamma_0 = \left[\left[\frac{1}{r}, 1, 0, {\bf 0}\right], 
\Theta_0\right]$.
At each time slot $k$, we have the following set of controls $\{0,1\}$
where 0 represents ``{\sf take another sample}'', and 1 represents
``{\sf stop and declare change}''.  Thus, at time slot $k$, when the
control chosen is 1, the state $\Gamma_{k+1}$ is given by a terminal
absorbing state ${\sf t}$;  when the control chosen is 0, the state
evolution is given by $\Gamma_{k+1}  =  [{\bf Q}_{k+1}, 
{\bf \Theta}_{k+1}]$, where 
\begin{eqnarray}
\label{eqn:gamma_evolution}
{\bf Q}_{k+1}          & = & {\bf \phi}_{{\bf Q}}({\bf Q}_k,M_k),\nn
{\bf\Theta}_{k+1}            
  & = & \left\{
	  \begin{array}{ll}
	  \left[\Theta_k + {\bf 1}_{\{T = k+1\}}\right], & \text{if} \ \Delta_{k+1}=0 \\
	   \left[\Theta_{k-\Delta_k}, \Theta_{k-\Delta_k+1},
\cdots, \Theta_k, \Theta_k + {\bf 1}_{\{T = k+1\}}\right], & \text{if} \
	  \Delta_{k+1}=\Delta_k+1 \\
	   \left[\Theta_{k-\Delta_k+\frac{1}{r}},
	   \Theta_{k-\Delta_k+\frac{1}{r}+1},
\cdots, \Theta_k, \Theta_k + {\bf 1}_{\{T = k+1\}}\right], & \text{if} \
	  \Delta_{k+1}=\Delta_k+1-\frac{1}{r}. 
	  \end{array}
	  \right.\nonumber\\
	  & =: & {\bf \phi}_{\bf \Theta}\left({\bf \Theta}_k,{\bf Q}_k,M_k,{\bf
	  1}_{\{T = k+1\}}\right)
\end{eqnarray}
where it is easy to observe that $\Theta_k + {\bf 1}_{\{T
= k+1\}} = \Theta_{k+1}$. When $\Delta_{k+1} = \Delta_k+1$, the 
batch $B_k$ is still in service, and hence, 
in addition to the current state $\Theta_{k+1} = 
\Theta_k + {\bf 1}_{\{T = k+1\}}$, 
we need to keep the states   
$\Theta_{k-\Delta_k}, \Theta_{k-\Delta_k+1},
\cdots, \Theta_k$. Also, when $\Delta_{k+1} =
\Delta_k+1-\frac{1}{r}$, then the batch index is incremented, and hence, 
the vector of states that determines the distribution of the observations 
sampled at or after $B_{k+1}/r$ and before $k+1$ is given by
$\left[\Theta_{k-\Delta_k+\frac{1}{r}}, \Theta_{k-\Delta_k+\frac{1}{r}+1},
\cdots, \Theta_k\right.$, $\left.\Theta_k + {\bf 1}_{\{T = k+1\}}\right]$. 

Define $O_k := {\bf 1}_{\{T = k+1\}}$, and define ${\bf N}_k := [M_k, O_k]$ 
be the state--noise during time slot $k$. 
The distribution of state--noise ${\bf N}_k$
given the state of the discrete--time system $\Gamma_k$ is given by $\prob{M_{k} = m, O_k = o
\big\arrowvert\Gamma_k = \left[{\bf q},{\mbox{\boldmath$\theta$}}\right]}$ and
is  the product of the distribution functions, 
$\prob{M_k = m   \big\arrowvert\Gamma_k = \left[{\bf
q},{\mbox{\boldmath$\theta$}}\right]}$ and 
$\prob{O_{k} = o \big\arrowvert\Gamma_k = \left[{\bf q},{\mbox{\boldmath$\theta$}}\right] 
}$. 
These distribution functions are provided in Appendix -- III. 

The problem is to detect the change in the state $\Theta_k$ as early as possible 
by sequentially observing the samples at the decision maker. 
\subsection{The NADM Change Detection Problem}
\label{subsec-the-NADM-change-detection-problem}
We now formulate the {\sf NADM} change--detection problem in which the observations 
from the sensor nodes are sent over a random access network to the fusion center 
and the fusion center processes the samples in the {\sf NADM} mode.

In Section~\ref{subsec:system-state-evolution-model}, we defined the
state $\Gamma_k = [{\bf Q}_k, {\bf \Theta}_k]$ on which we formulate
the {\sf NADM} change detection problem as a POMDP.
Recall that at the beginning of slot $k$, the decision maker receives a vector of sensor 
measurements ${\bf Y}_k$ and observes the state ${\bf Q}_k$ of the queueing 
system. Thus, at time $k$, ${\bf Z}_{k} = [{\bf Q}_k,{\bf Y}_k]$ is the 
observation of the decision maker about the state of the dynamical system $\Gamma_k$.

Let $A_k \in \{0,1\}$ be the control 
(or action) chosen by the decision maker after having observed ${\bf Z}_k$ 
at $k$. Recall that $0$ represents ``$\mathsf{take~another~sample}$'' and $1$ 
represents the action ``$\mathsf{stop~and~declare~change}$''. 
Let ${\bf I}_k = 
\left[{\bf Z}_{[0:k]},A_{[0:k-1]}\right]$ be the {\em information} 
vector\footnote{The notation ${\bf Z}_{[k_1:k_2]} := 
{\bf Z}_{k_1},{\bf Z}_{k_1+1},\cdots,{\bf Z}_{k_2}$} 
that is available to the decision maker, at the beginning of time slot $k$. 
Let $\tau$ be a stopping time with respect to the sequence of random variables 
${\bf I}_1, {\bf I}_2, \cdots$. 
Note that $A_k = 0$ for $k < \tau$ and $A_k = 1$ for $k \geqslant \tau$.
We are interested in obtaining a stopping time $\tau$ (with respect to
the sequence ${\bf I}_1,{\bf I}_2,\cdots$) that 
minimizes the mean detection delay subject to a constraint on the probability 
of false alarm. 
\begin{eqnarray}
 \min & & \EXP{(\tau -T)^+} \\
\text{such that}& & \prob{\tau < T} \leqslant \alpha \nonumber
\label{eqn:na-cdp}
\end{eqnarray}
Note that in the case of NADM, at any time $k$, a decision about the 
change is made based on the information ${\bf I}_k$ (which includes the
batch index we are processing and the delays). Thus, in the case of {\sf
NADM}, false alarm is defined as the event $\{\tau < T\}$ and, hence,
$\tau \geqslant T$ is not classified as a false alarm even if it is due to
pre--change measurements only. However, in the case of {\sf NODM},
this is classified as a false alarm as the
decision about the change is based on the batches received
until time $k$. 

Let $c$ be the cost per unit delay in detection. We are interested in 
obtaining a stopping time $\tau^*$ that minimizes the expected cost
(Bayesian risk) given by
\begin{eqnarray}
\label{eqn:na-opt-prob-middle} 
 C(c,\tau^*) 
& = & \min_\tau \EXP{{\bf 1}_{\{\Theta_\tau = 0\}} + c \cdot (\tau-T)^+ } \nn
& = & \min_\tau \EXP{{\bf 1}_{\{\Theta_\tau = 0\}} 
       + c \cdot \sum_{k=0}^{\tau-1} {\bf 1}_{\{\Theta_k = 1\}}  } \nn
& = & \min_\tau \EXP{g_\tau(\Gamma_\tau,A_\tau) +  \sum_{k=0}^{\tau-1} g_k(\Gamma_k,A_k)  }\nn 
& = & \min_\tau \EXP{ \sum_{k=0}^{\infty} g_k(\Gamma_k,A_k)}
\end{eqnarray}
where, as defined earlier, $\Gamma_k = [{\bf Q}_k, {\bf \Theta}_k]$.
Let ${\mbox{\boldmath$\theta$}} = [\theta_\delta,
\theta_{\delta-1},\cdots,\theta_1,\theta_0]$. We define for $k \leqslant
\tau$
\begin{eqnarray}
\label{eqn:g-defn}
g_k([{\bf q}, {\mbox{\boldmath$\theta$}}] ,a) 
& = & 
\left\{
\begin{array}{ll}
    0, & \text{if} \ \theta_0 = 0,  a = 0\\
    1, & \text{if} \ \theta_0 = 0,  a = 1\\
    c, & \text{if} \ \theta_0 = 1,  a = 0\\
    0, & \text{if} \ \theta_0 = 1,  a = 1
\end{array}
\right.
\end{eqnarray}
and for $k > \tau$, $g_k(\cdot, \cdot) :=0$. 
Recall that $A_k = 0$ for $k < \tau$ and $A_k = 1$ for $k \geqslant \tau$.
Note that $A_k$, the control at time slot $k$, depends only on ${\bf I}_k$.
Thus, every stopping time $\tau$, corresponds to a policy 
$\mu = (\mu_0,\mu_1,\cdots)$ such that 
$A_k = \mu_k({\bf I}_k)$, with $A_k = 0$ for $k < \tau$ and 
$A_k = 1$ for $k \geqslant \tau$. Thus,  
Eqn.~\ref{eqn:na-opt-prob-middle} can be written as  
\begin{eqnarray}
\label{eqn:na-opt-prob} 
 C(c,\tau^*) 
& = & \min_\mu \EXP{\sum_{k=0}^\infty g_k(\Gamma_k,A_k) }\nn
& = & \min_\mu \sum_{k=0}^\infty \EXP{ g_k(\Gamma_k,A_k) } \ \text{(by
monotone convergence theorem)} \nn
& = & \min_\mu \sum_{k=0}^\infty \EXP{ g_k(\Gamma_k,\mu_k({\bf I}_k)) }
\end{eqnarray}
Since ${\bf \Theta}_k$ is observed
only through ${\bf I}_k$, we look at a sufficient statistic for ${\bf I}_k$ in the next
subsection.


\subsection{Sufficient Statistic} 
\label{subsec:sufficient-statistic} 
In Section~\ref{subsec:evol-Q}, we have illustrated the evolution of the queueing 
system ${\bf Q}_k$ and we have shown in different scenarios, the vector ${\bf Y}_k$ 
received by the decision maker. Recall from Section~\ref{subsec:evol-Q} that
\begin{align*}
{\bf Y}_{k+1} = & \left\{
\begin{array}{ll}
 \emptyset, & \ \text{if} \ M_k = 0,\\
 \emptyset, & \ \text{if} \ M_k = j > 0, R_k^{(j)} = 1,\\ 
 Y_{k+1,0}, & \ \text{if} \ M_k = j > 0, R_k^{(j)} = 0, {\mysum i 1 N} R_k^{(i)} < N-1\\ 
 \left[Y_{k+1,0},Y_{k+1,1},\cdots,Y_{k+1,n}\right], & \ \text{if} \ M_k = j > 0, R_k^{(j)} = 0,
{\mysum i 1 N} R_k^{(i)} = N-1,\\
& \ \ \ {\mysum i 1 N} {\bf 1}_{\{W_k^{(i)} > 0\}} = n. 
\end{array}
\right.
\end{align*}
Note that $Y_{k+1,0}$ corresponds to $X_{B_k}^{(M_k)}$. The last part of
the above equation corresponds to the last pending sample of the batch
$B_k$ arriving at the decision maker at time $k+1$, with some samples
from batch $B_k+1$ $(=B_{k+1})$ also being released by the sequencer.
In this case, the state of nature at the sampling instant of the batch 
$B_{k+1}=B_k+1$ is $\Theta_{k-\Delta_k+1/r}$. Note that 
$\Theta_{k-\Delta_k+1/r}$ is a component of the vector ${\bf \Theta}_k$
as $k-\Delta_k+1/r = (B_k+1)/r < k$.
Thus, the distribution of 
$Y_{k+1,0},Y_{k+1,1},\cdots,Y_{k+1,n}$ is given by
\begin{align*}
f_{Y_{k+1,0}}(\cdot) = & \left\{
\begin{array}{ll}
f_0(\cdot), & \text{if} \  \Theta_{k-\Delta_k} = 0 \\
f_1(\cdot), & \text{if} \  \Theta_{k-\Delta_k} = 1 \ \text{and}
\end{array}
\right. \\
f_{Y_{k+1,i}}(\cdot) = & \left\{
\begin{array}{ll}
f_0(\cdot), & \text{if} \  \Theta_{k-\Delta_k+1/r}  = 0 \\
f_1(\cdot), & \text{if} \  \Theta_{k-\Delta_k+1/r}  = 1
\end{array}
\right., \ i=1,2,\cdots,n. 
\end{align*}
Thus, at time $k+1$, the current observation ${\bf Y}_{k+1}$ depends only on the 
previous state $\Gamma_k$, previous action $A_k$, and the previous noise 
of the system ${\bf N}_k$. Thus, a sufficient statistic is  
$\left[\prob{\Gamma_k=\left[{\bf q}, {\mbox{\boldmath$\theta$}}\right] 
\big\arrowvert {\bf I}_k}\right]_{[{\bf q}, {\mbox{\boldmath$\theta$}} ]\in\mathcal{S}}$ 
(see page 244, \cite{books.bertsekas00a}) where $\mathcal{S}$ is the 
set of all states of the dynamical system defined in 
Sec.~\ref{subsec:system-state-evolution-model}. 
Let ${\bf q} = [\lambda,b,\delta,{\bf w},{\bf r}]$.
Note that 
{
\begin{eqnarray}
\label{eqn:sufficient-stat}
& &\prob{\Gamma_k = \left[{\bf q},
{\mbox{\boldmath$\theta$}}\right] \big\arrowvert{\bf I}_k}\nn   
& = &\prob{\Gamma_k = \left[{\bf q},
{\mbox{\boldmath$\theta$}}\right] \big\arrowvert{\bf I}_{k-1}, {\bf
Q}_k, {\bf Y}_k}\nn   
& = & {\bf 1}_{\{{\bf Q}_k = {\bf q} \}} \cdot 
\prob{{\bf \Theta}_k = {\mbox{\boldmath$\theta$}}\big\arrowvert {\bf
I}_{k-1}, {\bf Q}_k = {\bf q}, {\bf Y}_k } \nn 
& = & {\bf 1}_{\{{\bf Q}_k = {\bf q} \}}\nn
& & \cdot 
 \prob{\left[\Theta_{k-\delta},
\Theta_{k-\delta+1},\cdots,\Theta_{k-1},\Theta_k\right] = 
\left[\theta_{\delta}, \theta_{\delta-1}, \cdots, \theta_1,
\theta_0\right] \big\arrowvert {\bf
I}_{k-1}, {\bf Q}_k = {\bf q},{\bf Y}_k } \nn 
& = & {\bf 1}_{\{{\bf Q}_k = {\bf q} \}} \cdot 
\prob{\Theta_{k-\delta} = \theta_{\delta}
\big\arrowvert {\bf I}_{k-1}, {\bf Q}_k = {\bf q}, {\bf Y}_k
} \nn 
& & \cdot \prod_{j=1}^\delta \prob{\Theta_{k-\delta+j} = \theta_{\delta-j}
\big\arrowvert \Theta_{k-\delta+j'}=\theta_{\delta-j'},
j'=0,1,\cdots,j-1, {\bf I}_{k-1}, {\bf Q}_k = {\bf q}, {\bf Y}_k }\nn 
\end{eqnarray}
}
Observe that
\begin{eqnarray*}
 & & \prob{\Theta_{k-\delta+j} = \theta_{\delta-j} 
\big\arrowvert \Theta_{[k-\delta:k-\delta+j-2]},  \Theta_{k-\delta+j-1}=
0, {\bf I}_{k-1}, {\bf Q}_k = {\bf q},
{\bf Y}_k }\nn 
& = & \left\{
\begin{array}{ll}
1-p, & \text{if} \ \theta_{\delta-j} = 0\\
p,   & \text{if} \ \theta_{\delta-j} = 1
\end{array}
\right.
\end{eqnarray*}
and
\begin{eqnarray*}
&&  \prob{\Theta_{k-\delta+j} = \theta_{\delta-j} 
\big\arrowvert \Theta_{[k-\delta:k-\delta+j-2]},
\Theta_{k-\delta+j-1}=1,  {\bf I}_{k-1}, {\bf Q}_k =
{\bf q},
{\bf Y}_k }\nn 
& = & \left\{
\begin{array}{ll}
0, & \text{if} \ \theta_{\delta-j} = 0\\
1,   & \text{if} \ \theta_{\delta-j} = 1.
\end{array}
\right.
\end{eqnarray*}
This is because given $\Theta_{k-\delta} $, 
the events 
$\{\Theta_{k-\delta+j} = \theta_{\delta-j} \}$, 
$\left\{{\bf I}_{k-1}, {\bf Q}_k = {\bf q}, 
{\bf Y}_k
\right\}$ are conditionally independent. 
Thus, Eqn.~\ref{eqn:sufficient-stat} can be written as
{
\begin{eqnarray}
\label{eqn:psi-suff-stat}
& &\prob{\Gamma_k = \left[{\bf q}, 
{\mbox{\boldmath$\theta$}}\right] \big\arrowvert{\bf I}_k}\nn   
& = & \left\{
\begin{array}{ll}
{\bf 1}_{\{{\bf Q}_k = {\bf q}\}} \cdot 
\prob{\Theta_{k-\delta} = 1 
\big\arrowvert {\bf I}_{k-1}, {\bf Q}_k = {\bf q}, {\bf Y}_k
}, & \text{if} \  {\mbox{\boldmath$\theta$}} = {\bf 1}\\
{\bf 1}_{\{{\bf Q}_k = {\bf q} \}} \cdot 
\prob{\Theta_{k-\delta} = 0 
\big\arrowvert {\bf I}_{k-1}, {\bf Q}_k = {\bf q}, {\bf Y}_k
} \cdot (1-p)^{\delta-j-1} p, & \text{if} \  {\mbox{\boldmath$\theta$}} =
[0,\cdots,0,\underbrace{1}_{\theta_{j}}, \cdots,
1]\\
{\bf 1}_{\{{\bf Q}_k = {\bf q} \}} \cdot 
\prob{\Theta_{k-\delta} = 0 
\big\arrowvert {\bf I}_{k-1}, {\bf Q}_k = {\bf q}, {\bf Y}_k
} \cdot (1-p)^{\delta}, & \text{if} \  {\mbox{\boldmath$\theta$}}
= {\bf 0}
\end{array}
\right.\nn
\end{eqnarray}
}
Define
$\widetilde{\Theta}_k := \Theta_{k-\Delta_k}$,
and define
\begin{eqnarray}
\label{eqn:pi}
\Psi_k & := & \prob{ {\widetilde{\Theta}_k} = 1 \big\arrowvert {\bf
I}_{k-1}, {\bf Q}_k = [\lambda,b,\delta,{\bf w},{\bf r}],{\bf Y}_k 
}\nn 
&= & \prob{ \Theta_{k - \delta}=1 \big\arrowvert {\bf I}_{k-1}, {\bf
Q}_k = [\lambda,b,\delta,{\bf w},{\bf r}],{\bf Y}_k  }\nn 
\Pi_k  & := & \prob{ {\Theta_k} = 1 \big\arrowvert {\bf I}_{k-1}, {\bf
Q}_k = [\lambda,b,\delta,{\bf w},{\bf r}],{\bf Y}_k   } \nn
& = &
\prob{ T \leqslant k \big\arrowvert {\bf I}_{k-1}, {\bf Q}_k =
[\lambda,b,\delta,{\bf w},{\bf r}], {\bf Y}_k  }. \nn 
\end{eqnarray}
Thus, Eqn.~\ref{eqn:psi-suff-stat} can be written as 
\begin{eqnarray}
\label{eqn:1psi-suff-stat}
& &\prob{\Gamma_k = \left[[\lambda,b,\delta,{\bf w},{\bf r}] ,
{\mbox{\boldmath$\theta$}}\right] \big\arrowvert{\bf I}_k}\nn   
& = & \left\{
\begin{array}{ll}
{\bf 1}_{\{{\bf Q}_k = [\lambda,b,\delta,{\bf w},{\bf r}] \}} \cdot
\Psi_k, & \text{if} \  {\mbox{\boldmath$\theta$}} = {\bf 1}\\
{\bf 1}_{\{{\bf Q}_k = [\lambda,b,\delta,{\bf w},{\bf r}] \}} \cdot 
(1-\Psi_k) \cdot (1-p)^{\delta-j-1} p, & \text{if} \  {\mbox{\boldmath$\theta$}} =
[0,\cdots,0,\underbrace{1}_{\theta_{j}}, \cdots,
1]\\
{\bf 1}_{\{{\bf Q}_k = [\lambda,b,\delta,{\bf w},{\bf r}] \}} \cdot 
(1-\Psi_k) \cdot (1-p)^{\delta}, & \text{if} \  {\mbox{\boldmath$\theta$}}
= {\bf 0}
\end{array}
\right.
\end{eqnarray}

We now find a relation between $\Pi_k$ and $\Psi_k$ in the following
Lemma.
\begin{lemma}
\label{lemma}
The relation between the conditional probabilities 
$\Pi_k$ and $\Psi_k$ is given by
\begin{eqnarray}
\label{eqn:psi-theta}
\Pi_k & = & \hspace{-0mm} \Psi_k + (1-\Psi_k)\left(1 - (1-p)^{\delta}\right)
\end{eqnarray}
\end{lemma}

\begin{proof}
See Appendix -- IV.
\end{proof}

From 
Eqn.~\ref{eqn:1psi-suff-stat} and
Lemma~\ref{lemma}, it is clear that a sufficient statistic for 
${\bf I}_k$ is $\nu_k = [{\bf Q}_k, \Pi_k]$. Also, we show in 
Appendix -- V that $\nu_k$ can be computed recursively,
i.e., when $A_k = 0$, $\nu_{k+1} = \left[{\bf Q}_{k+1}, \Pi_{k+1}\right]
= \left[{\bf Q}_{k+1}, \phi_{\Pi}(\nu_k,{\bf Z}_{k+1})\right]$, 
and when $A_k = 1$, $\nu_{k+1} = {\sf t}$, a terminal state. 
Thus, $\nu_k$ can be thought of as entering into a terminating (absorbing) state 
${\sf t}$ at $\tau$ (i.e., $\nu_k = [{\bf Q}_k, \Pi_k]$ for $k < \tau$ and 
$\nu_k = {\sf t}$ for $k \geqslant \tau$). Since $\nu_k$ is sufficient, 
for every policy $\mu_k$ there corresponds a policy 
$\widetilde{\mu}_k$ such that  
$\mu_k({\bf I}_k) = \widetilde{\mu}_k(\nu_k)$ 
(see page 244, \cite{books.bertsekas00a}).

\subsection{Optimal Stopping Time $\tau$}
\label{subsec:optimal-stopping-time}
Let $\mathcal{Q}$ be the set of all possible states of the queueing system, ${\bf Q}_k$. 
Thus the state space of the sufficient statistic is $\mathcal{N} = 
\left(\mathcal{Q}\times[0,\ 1]\right)\cup\{{\sf t}\}$. Recall that the action
space is $\mathcal{A} = \{0,1\}$. Define the one--stage cost function 
$\widetilde{g}:\mathcal{N} \times \mathcal{A} \to \mathbb{R}_+$ as follows. Let 
$\nu \in \mathcal{N}$ be a state of the system and let $a \in \mathcal{A}$ 
be a control. Then,
\begin{eqnarray*} 
\widetilde{g}(\nu,a) & = & 
\left\{
\begin{array}{ll}
    0 & \text{if} \ \nu =    {\sf t}\\
  c\cdot \pi & \text{if} \ \nu = [{\bf q}, \pi], a = 0\\
1-\pi & \text{if} \ \nu = [{\bf q}, \pi], a = 1.
\end{array}
\right.
\end{eqnarray*}
Note from Eqn.~\ref{eqn:g-defn} for $k \leqslant \tau$ that 
\begin{eqnarray*}
\EXP{g_k(\Theta_k,A_k)}
 & = & \EXP{g_k(\Theta_k,\mu_k({\bf I}_k))}\\ 
 & = & \EXP{\EXP{g_k(\Theta_k,\mu_k({\bf I}_k)) \bigg\arrowvert {\bf I}_k}}\\ 
 & = & \EXP{\widetilde{g}(\nu_k,\widetilde{\mu}_k(\nu_k)) } 
\end{eqnarray*}
and for $k > \tau$,
\begin{eqnarray*}
\EXP{g_k(\Theta_k,A_k)}
& = & 0 \nn 
& = & \EXP{\widetilde{g}({\sf t},\cdot)} 
\end{eqnarray*}
Since, $\{\nu_k\}$ is a controlled Markov process, and the one--stage cost
function $\widetilde{g}(\cdot,\cdot)$, the transition probability kernel for $A_k=1$
and for $A_k = 0$ (i.e., $\prob{{\bf Z}_{k+1}\big\arrowvert \nu_k}$), do not 
depend on time $k$, and the optimization problem defined in Eqn.~\ref{eqn:na-opt-prob} 
is over infinite horizon, it is sufficient 
to look for an optimal policy in the 
space of stationary Markov policies 
(see page 83, \cite{books.bertsekas00b}).
Thus, the optimization problem defined in Eqn.~\ref{eqn:na-opt-prob} can be written as 
\begin{eqnarray} 
 C(c,\tau^*) & = &  \min_{\widetilde{\mu}}\sum_{k=0}^\infty {\mathsf E}\left[\widetilde{g}\big(\nu_k,\widetilde{\mu}_k(\nu_k)\big)\right]\nn
  & = &  \sum_{k=0}^\infty {\mathsf E}\left[\widetilde{g}\big(\nu_k,\widetilde{\mu}^*(\nu_k)\big)\right]. 
\end{eqnarray}
Thus, the optimal total cost is given by 
\begin{eqnarray} 
\label{eqn:need-to-solve} 
J^*([{\bf q}_0,\pi_0]) & = & 
    \sum_{k=0}^\infty {\mathsf E}\left[\widetilde{g}\big(\nu_k,\widetilde{\mu}^*(\nu_k)\big) \bigg\arrowvert \nu_0 = [{\bf q}_0,\pi_0] \right]. 
\end{eqnarray}
The solution to the above problem is obtained following the Bellman's equation, 
\begin{eqnarray} 
\label{eqn:bellman-sits} 
J^*([{\bf q}, \pi]) & := & \min\left\{1-\pi, c\pi + {\mathsf E}\left[
J^*\left({\bf Q}_{k+1},\phi_{\Pi}(\nu_k, {\bf Z}_{k+1})\right) \bigg\arrowvert \nu_k = [{\bf q},\pi] \right] \right\}.\nn
\end{eqnarray} 
where the function $\phi_{\Pi}(\nu_k, {\bf Z}_{k+1})$ is
provided in Appendix -- V.

\begin{remarks}
\label{thm:decision_queueing_coupling}
The optimal stationary Markov policy (i.e., the optimum stopping 
rule $\tau$) in general depends on ${\bf Q}$. Hence, the decision 
delay and the queueing delay are coupled, unlike in the {\sf NODM} case. 
\end{remarks}

\noindent
We characterize the optimal policy in the following theorem.
\begin{theorem}
The optimal stopping rule $\tau^*$ is a network--state dependent 
threshold rule on the a posteriori probability $\Pi_k$, i.e.,
there exist thresholds $\gamma({\bf q})$ such that
\begin{eqnarray} 
\tau & = & \inf\{k \geqslant 0 : \Pi_k \geqslant \gamma({\bf Q}_k)\}
\end{eqnarray} 
\end{theorem}
\begin{proof}
See Appendix--VI.
\end{proof}

In general, the thresholds $\gamma({\bf Q}_k)$s (i.e., optimum policy) 
can be numerically obtained by solving 
Eqn.~\ref{eqn:bellman-sits} using value iteration method 
(see pp. 88--90, \cite{books.bertsekas00b}). However, computing the optimal policy for the
{\sf NADM} procedure is hard as the state space is huge even for moderate 
values of $N$. Hence, we resort to a suboptimal policy based on the 
following threshold rule, which is motivated by the structure of the 
optimal policy.    
\begin{eqnarray}
\label{eqn:nadm-subopt}
\tau & = & \inf\{k \geqslant 0 : \Pi_k \geqslant \gamma\}
\end{eqnarray}
where $\gamma$ is chosen such that $\prob{\tau < T} = \alpha$ is met.

Thus, we have formulated a sequential change detection problem when the sensor 
observations are sent to the decision maker over a random access network, and 
the fusion center processes the samples in the {\sf NADM} mode. The information 
for decision making now needs to include the network state ${\bf Q}_k$ (in 
addition to the samples received by the decision maker); we have shown that 
$[{\bf Q}_k,\Pi_k]$ is sufficient for the {\em information} history ${\bf I}_k$. 
Also, we have provided the structure for the optimal policy. Since, obtaining 
the optimal policy is computationally hard, we gave a simple threshold based 
policy, which is motivated by the structure of the optimal policy. 

\section{Numerical Results}
\label{sec:optimal_parameters}
Minimizing the mean detection delay not only requires an optimal decision rule 
at the fusion center but also involves choosing the optimal values of the 
sampling rate $r$, and the number of sensors $N$. To explore this, we obtain the 
minimum decision delay for each value of the sampling rate $r$ numerically, and 
the network delay via simulation.

\subsection{Optimal Sampling Rate}
\label{sec:optimal_sampling_rate}
Consider a sensor network with $N$ nodes. For a given probability of false alarm, 
the decision delay (detection delay without the network--delay component) 
decreases with increase in sampling rate. This is due to the increase in the 
number of samples that the fusion center receives within a given time. But, as 
the sampling rate increases, the network delay increases due to the increased 
packet communication load in the network. Therefore it is natural to expect the 
existence of a sampling rate $r^*$, with $r^*<\sigma/N$, (the sampling rate should 
be less than $\sigma/N$, for the queues to be stable; see 
Theorem~\ref{thm:fjq-gps_stationary_delay}) that optimizes the tradeoff between 
these two components of detection delay. Such an $r^*$, in the case of {\sf NODM} 
can be obtained by minimizing the following expression over $r$ (recall 
Theorem~\ref{thm:decoupling}).
\[
\left(d(r)+l(r)\right)(1-\alpha) -\rho \cdot l(r) + \frac{1}{r}\min_{\Pi_\alpha} \EXP{\widetilde{K}-K}^+
\]
Note that in the above expression, the delay term $\min_{\Pi_\alpha} 
\EXP{\widetilde{K}-K}^+$ also depends on the sampling rate $r$ via the 
probability of change $p_r = 1 - (1-p)^{(1/r)}$. 
The delay due to coarse sampling $l(r)(1-\alpha) - \rho \cdot l(r)$
can be found analytically (see Appendix -- I). We can approximate the 
delay $\min_{\Pi_\alpha} \EXP{\widetilde{K}-K}^+$ by the asymptotic
(as $\alpha \to 0$) delay as
$\frac{|\ln(\alpha)|}{NI(f_1,f_0)+|\ln(1-p_r)|}$ where $I(f_1,f_0)$
is the Kullback--Leibler (KL) divergence between the pdfs $f_1$ and $f_0$ 
(see \cite{tartakovsky-veeravalli05general-asymptotic-quickest-change}).
But, obtaining the network--delay (i.e., $d(r)(1-\alpha)$) analytically 
is hard, and hence an analytical characterisation of $r^*$ appears 
intractable.  Hence, we have resorted to numerical evaluation. 
\begin{figure}[t]
   \centering \ 
   \psfig{figure=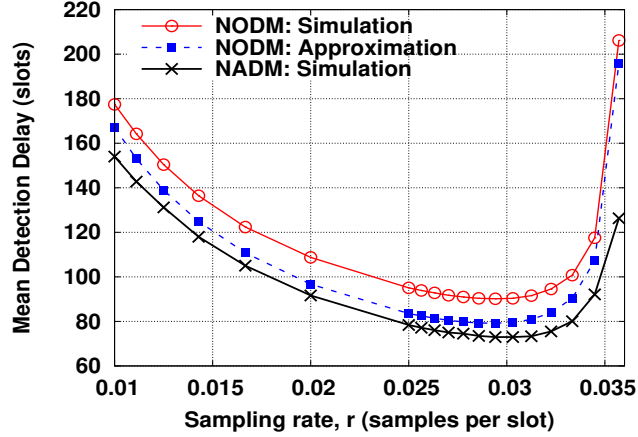,height=6cm,width=9cm}
   \caption{Mean detection delay for $N=10$ nodes is plotted against the 
            sampling rate $r$ for both {\sf NODM} and {\sf NADM} (defined 
            in Eqn.~\ref{eqn:nadm-subopt}). For {\sf NODM},
            an approximate analysis is also plotted. This was obtained with the 
            prior probability $\rho=0,~ p=0.0005$, probability of false alarm 
            target $\alpha=0.01,~\sigma=0.3636$ and with the sensor observations
            being ${\cal{N}}(0,1)$ and ${\cal{N}}(1,1)$, before and after the
            change respectively.}
   \label{fig:optimal-sampling-rate}
\end{figure} 

The distribution of sensor observations are taken to be ${\cal{N}}(0,1)$ and 
${\cal{N}}(1,1)$\footnote{As usual, ${\cal{N}}(a,v)$ denotes a normal 
distribution with mean $a$ and variance $v$}, before and after the change 
respectively for all the $10$ nodes. We choose the probability of occurrence 
of change in a slot to be $p=0.0005$, i.e., the mean time until change is 
$2000$ slots. $\min_{\Pi_\alpha}\EXP{\widetilde{K}-K}^+$ and $d(r)$ are obtained 
from simulation for $\alpha=0.01$ and $\sigma=0.3636$ and the expression for mean 
detection delay (displayed above) is plotted against $r$ 
in Figure~\ref{fig:optimal-sampling-rate}. Note that both NODM and NADM
are threshold based, and we obtain the corresponding thresholds for 
a target ${\sf P_{FA}} = 0.01$ by simulation. These thresholds are then used
to obtain the mean detection delay by simulation. 
In Figure~\ref{fig:optimal-sampling-rate}, 
we also plot the approximate mean detection delay which is obtained through the 
expression for $l(r)$ and the approximation,
$\min_{\Pi_\alpha} \EXP{\widetilde{K}-K}^+ \approx \frac{|\ln(\alpha)|}{NI(f_1,f_0)+|\ln(1-p_r)|}$. 
We study this approximation as this provides an (approximate) explicit expression 
for the mean decision delay. The delay in the FJQ--GPS does not have a closed form 
expression. Hence, we still need simulation for the delay due to queueing network.
It is to be noted that
at $k=0$, the size of all the queues is set to 0. 
The mean detection delay due to the procedure defined in Eqn.~\ref{eqn:nadm-subopt}
is also plotted in Figure~\ref{fig:optimal-sampling-rate}. 

As would have been expected, we see from Figure~\ref{fig:optimal-sampling-rate} 
that the {\sf NADM} procedure has a better mean detection delay performance than 
the {\sf NODM} procedure. Note that $\sigma/N = 0.03636$ and hence for the queues 
to be stable (see Theorem~\ref{thm:fjq-gps_stationary_delay}), the sampling rate
has to be less that $\sigma/N = 0.03636$ ($1/28 < 0.03636 < 1/27$). As the sampling 
rate $r$ increases to 1/28 (the maximum allowed sampling rate), the queueing delay 
increases rapidly. This is evident from Figure~\ref{fig:optimal-sampling-rate}. Also, 
we see from Figure~\ref{fig:optimal-sampling-rate} that operating at a sampling rate 
around $1/34 (\approx 0.0294)$ samples/slot would be optimum. The optimal sampling 
rate is found to be approximately the same for {\sf NODM} and {\sf NADM}. At the 
optimal sampling rate the mean detection delay of {\sf NODM} is 90 slots and that 
of {\sf NADM} is 73 slots.

\subsection{Optimal Number of Sensor Nodes (Fixed Observation Rate)}

Now let us consider fixing $N\times r$. This is the number of observations the 
fusion center receives per slot in a network with $N$ nodes sampling at a rate 
$r$ (samples per slot). It is also a measure of the energy spent by the network 
per slot. Since it has been assumed that the observations are conditionally 
independent and identically distributed across the sensors and over time, it is 
natural to ask how beneficial it is to have more nodes sampling at a lower rate, 
when compared to fewer nodes sampling at a higher rate with the number of 
observations per slot being the same. With $p=0.0005$, $\alpha=0.01$, and
$\sigma=0.3636$, and $f_0\sim\mathcal{N}(0,1)$ and $f_1\sim\mathcal{N}(1,1)$, we 
present simulation results for two examples, the first one being $Nr = 1/3$ (the 
case of heavily loaded network) and the second one being $Nr = 1/100$ (the case 
of lightly loaded network, $Nr \ll \sigma$).

\begin{figure}[t]
   \begin{center}
   \begin{minipage}{3.8cm}
   \begin{center}
    \psfig{figure=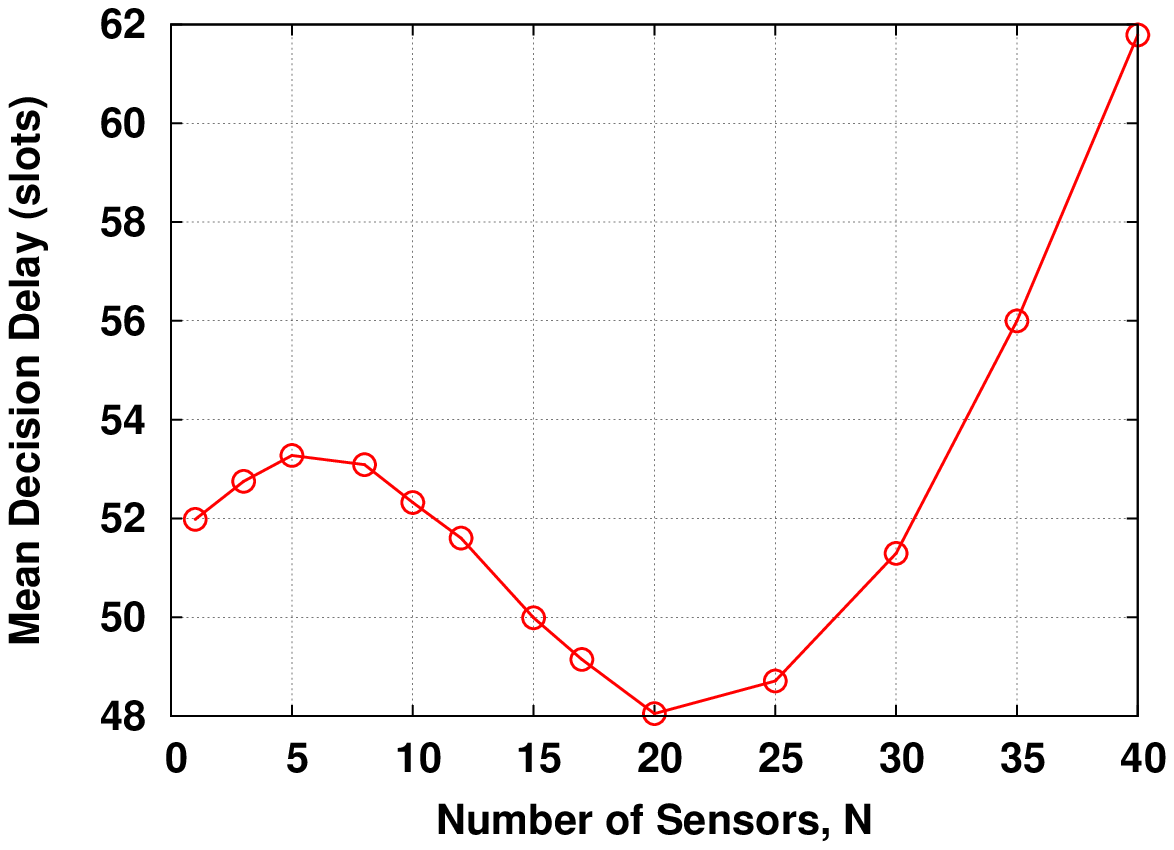,height=4cm,width=5.5cm}
   \end{center}
   \end{minipage}
\hspace{25mm}
   \begin{minipage}{3.8cm}
   \begin{center}
    \psfig{figure=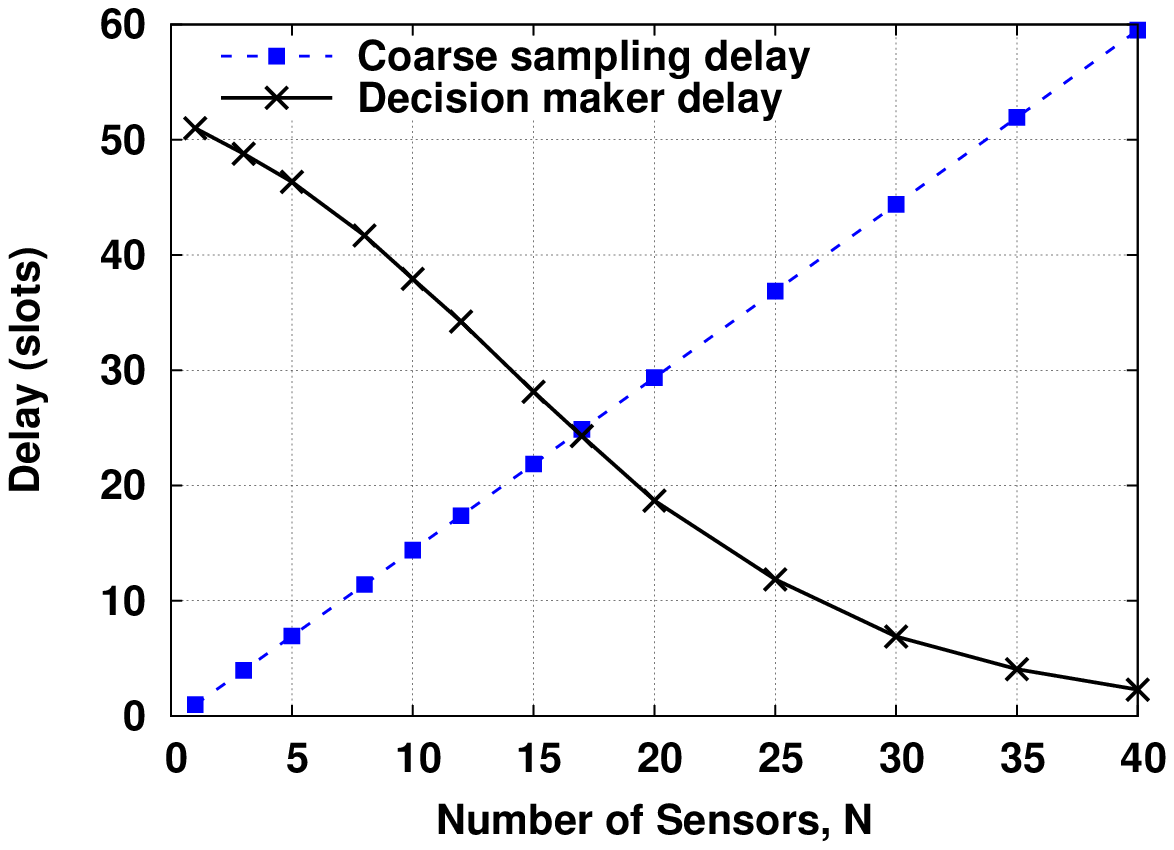,height=4cm,width=5.5cm}
   \end{center}
   \end{minipage}
  \caption{Mean \emph{decision} delay of {\sf NODM} procedure for $N\times r=1/3$ is plotted
    against the the number of nodes $N$. The plot is obtained with 
    $\rho=0,~p=0.0005,~\alpha=0.01$ and with the sensor observations
    being ${\cal{N}}(0,1)$ and ${\cal{N}}(1,1)$, before and after the
    change respectively. The components of the mean decision delay, i.e., the coarse sampling delay 
    $(1-\alpha)l(r) - \rho l(r)$, and the
    decision maker delay, $\frac{1}{r}\min_{\Pi_\alpha} \EXP{\widetilde{K}-K}^+$ are shown on the right.}
  \label{fig:decision-delay}
  \end{center}
 \end{figure}

Figure~\ref{fig:decision-delay} shows the plot of mean decision delay, $l(r)(1-\alpha-\rho) 
 + \frac{1}{r}\min_{\Pi_\alpha}\EXP{\widetilde{K}-K}^+$ versus the 
number of sensors when $Nr = 1/3$. As $N$ increases, the sampling rate $r = 1/(3N)$ 
decreases and hence the coarse sampling delay $l(r) (1-\alpha)$ increases; this 
can be seem to be approximately linear by analysis of the expression for $l(r)$ given
in Appendix -- I. Also, 
as $N$ increases, the decision maker gets more samples at the decision instants and 
hence the delay due to the decision maker 
$\frac{1}{r}\min_{\Pi_\alpha}\EXP{\widetilde{K}-K}^+$ decreases (this is evident 
from the right side of Figure~\ref{fig:decision-delay}). Figure~\ref{fig:decision-delay} 
shows that in the region where $N$ is large (i.e., $N \geqslant 20$) 
or $N$ is very small (i.e., $N < 5$), as $N$ increases, the mean 
decision delay increases. This is because in this region as $N$ increases, the decrease 
in the delay due to decision maker is smaller compared to the increase in the delay due 
to coarse sampling. However, in the region where $N$ is moderate (i.e., for $5 \leqslant N < 20$), as $N$ 
increases, the decrease in the delay due to decision maker is large compared to the 
increase in the delay due to coarse sampling. Hence in this region, the mean decision 
delay decreases with $N$. Therefore, we infer that when $N\times r =\frac{1}{3}$,
deploying $20$ nodes sampling at $1/60$ samples per slot is optimal, when there is no 
network delay.

\begin{figure}[t]
   \centering \ 
   \psfig{figure=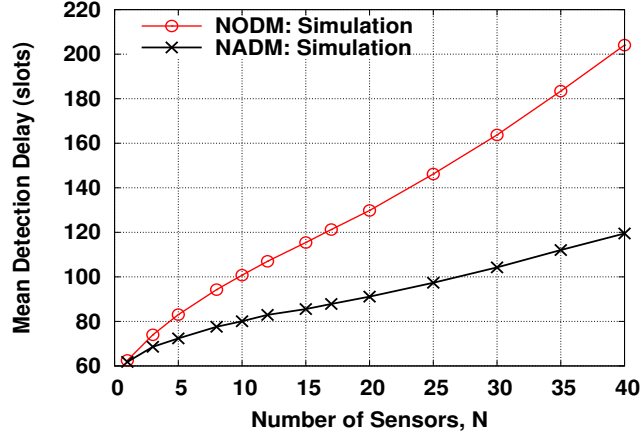,height=6cm,width=9cm}
   \caption{Mean \emph{detection} delay for $N\times r=1/3$ is plotted
     against the number of nodes $N$. This was obtained with
     $\rho=0,~p=0.0005,~\alpha=0.01~\sigma=0.3636$ and with the sensor
     observations being ${\cal{N}}(0,1)$ and ${\cal{N}}(1,1)$, before
     and after the change respectively.}
   \label{fig:detection-delay}
 \end{figure}

Figure~\ref{fig:detection-delay} shows the mean detection delay (i.e., the network 
delay plus the decision delay shown in Figure~\ref{fig:decision-delay}) versus the 
number of nodes $N$ for a fixed $N \times r = 1/3$. As the the number of nodes $N$ 
increases, the sampling rate $r = 1/(3N)$ decreases. For large $N$ (and equivalently 
small $r$), in the case of {\sf NODM} with the Shiryaev procedure, the network delay, 
$d(r) \approx \frac{N}{\sigma}$ as it requires $N$ (independent) successes, each with 
probability $\sigma$, in the random access network to transport a batch of $N$ samples 
(also, since the sampling rate $r$ is small, one would expect that a batch is delivered 
before a new batch is generated) and the decision maker requires just one batch of $N$ 
samples to stop (after the change occurs). Hence, for large $N$, the detection delay is 
$\approx l(r)(1-\alpha) + d(r)(1-\alpha) \approx l(r) (1-\alpha)+ \frac{N}{\sigma}(1-\alpha)$.
It is to be noted that for large $N$, to achieve a false 
alarm probability of $\alpha$, the decision maker requires $N_{\alpha} < N$ samples
(the mean of the log--likelihood ratio, LLR of received samples, after change, is the 
KL divergence between pdfs $f_1$ and $f_0$, given by $I(f_1,f_0) > 0$.
Hence, the posterior probability, which is a function of LLR, increases with the the number of received samples. Thus, 
to cross a threshold of $\gamma(\alpha)$, we need $N_\alpha$ samples). Thus, for large 
$N$, in the {\sf NADM} procedure, the detection delay is approximately 
$l(r)(1-\alpha) + \frac{N_\alpha}{\sigma}(1-\alpha)$, where $N_\alpha/\sigma$ is the mean 
network--delay to transport $N_\alpha$ samples. Thus, for large $N$, the 
difference in the mean detection 
delay between {\sf NODM} and {\sf NADM} procedures is approximately 
$\frac{1 - \alpha}{\sigma}(N-N_\alpha)$. Note that $N_\alpha$ depends only on $\alpha$ 
and hence the quantity $\frac{1 - \alpha}{\sigma}(N-N_\alpha)$ increases with $N$. This 
behaviour is in agreement with Figure~\ref{fig:detection-delay}. Also, as 
$N \times r = 1/3$, we expect the network delay to be very large (as 1/3 is close to 
$\sigma = 0.3636$) and hence having a single node is optimal which is also evident from 
Figure~\ref{fig:detection-delay}. 

 \begin{figure}
   \centering \ 
   \psfig{figure=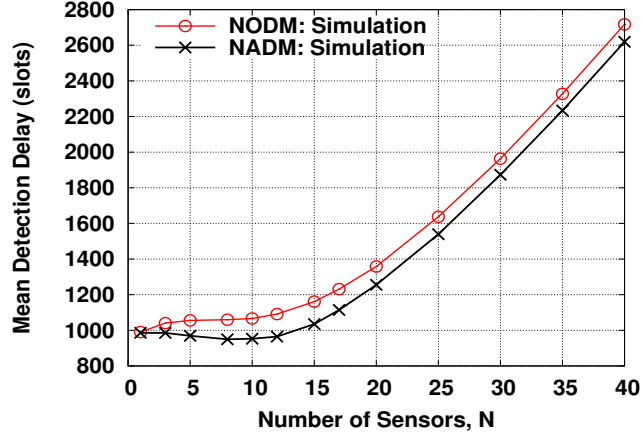,height=6cm,width=9cm}
   \caption{Mean \emph{detection} delay for $N\times r=0.01$ is plotted
     against the the number of nodes $N$. This was obtained with 
     $\rho=0,~p=0.0005,~\alpha=0.01$ and with the sensor observations
     being ${\cal{N}}(0,1)$ and ${\cal{N}}(1,1)$, before and after the
     change respectively.}
   \label{fig:decision-delay-01}
 \end{figure}
It is also possible to find an example where the optimal number of
nodes is greater than $1$. For example this occurs in the above
setting for $N \times r=0.01$ (see Figure~\ref{fig:decision-delay-01}).
Note that having $N=10$ sensors is optimal for the {\sf NADM} procedure. 
The {\sf NODM} procedure makes the decision only when it receives a batch
of $N$ samples corresponding to a sampling instant, whereas {\sf NADM}
procedure makes the decision at every time slot irrespective of 
whether it receives a sample in that time slot or not. Thus, the 
Bayesian update that {\sf NADM} does at every time slot makes it stop 
earlier than {\sf NODM}.

\section{Conclusions}
In this work we have considered the problem of minimizing the mean
detection delay in an event detection on a small extent ad hoc wireless
sensor network. We provide two ways of processing samples in the fusion
center: i) {\em Network Oblivious} ({\sf NODM}) processing, and ii) {\em
Network Aware} ({\sf NADM}) processing. We show that in the  {\sf NODM}
processing, under periodic sampling, the detection delay decouples into
decision and network delays. An important implication of this is that an
optimal sequential change detection algorithm can be used in the
decision device independently of the random access network. We also
formulate and solve the change detection problem in the {\sf NADM}
setting in which case the optimal decision maker needs to use the
network state in its optimal stopping rule. Also, we study the network
delay involved in this problem and show that it is important to operate
at a particular sampling rate to achieve the minimum detection delay.  

\section*{Appendix -- I}
\emph{Proof:} (Theorem~\ref{thm:decoupling})
\begin{eqnarray}
\label{eq:min_det_delay}
 \MPA \EXP{\left(\widetilde{U}-T\right)I_{\{\widetilde{T}\geqslant T\}}} 
 & = & \MPA \EXP{(\widetilde{U}-\widetilde{T} + \widetilde{T} - \frac{K}{r} + \frac{K}{r} - T ) I_{\{\widetilde{T}\geqslant T\}} }\nonumber\\
\label{eq:min_del_delay_one_two_three}
 & = & \MPA \left\{ 
                \EXP{(\widetilde{U}-\widetilde{T}) I_{\{\widetilde{T}\geqslant T\}}}
             +  \EXP{\left(\frac{K}{r} - T \right) I_{\{\widetilde{T}\geqslant T\}}}\right.\nonumber\\
& &             + \left. \frac{1}{r}\EXP{\left(\widetilde{K} - K\right) I_{\{\widetilde{T}\geqslant T\}}} \right\}
\end{eqnarray}
Note that in Eqn. \ref{eq:min_del_delay_one_two_three}, the first term is the 
queueing delay, the second term is the coarse sampling delay and the third
term is the decision delay (all delays being in slots). Consider the first term,
\begin{eqnarray*}
 \EXP{(\widetilde{U}-\widetilde{T})I_{\{\widetilde{T}\geqslant T\}}}
& = & \EXP{(U_{\widetilde{K}}-t_{\widetilde{K}})I_{\{\widetilde{T}\geqslant T\}}}\\
& = &\sum_{j\ge0,b\ge0,x\ge0}\prob{T=j, \widetilde{K}=b, D_b=x}x\cdot I_{\{\frac{b}{r}\geqslant j\}}\\
& = &\sum_{j\ge0,b\ge0,x\ge0}\prob{T=j, \widetilde{K}=b}\prob{D_b=x}x\cdot I_{\{\frac{b}{r}\geqslant j\}}
\end{eqnarray*}
where we have used the facts that (i) the decision process is based on only
what the packets carry and not on their arrival time etc, and (ii) the condition 
that sampling is done periodically at a known rate $r$. Assuming the queueing 
system to be stationary, the above can be written as
\begin{eqnarray*}
 \EXP{(\widetilde{U}-\widetilde{T})I_{\{\widetilde{T}\geqslant T\}}}
& = & \left(\sum_{x\ge0} \prob{D=x}x \right) \sum_{j,b}\prob{T=j, \widetilde{K}=b}I_{\{\frac{b}{r}\geqslant j\}}\\
& = & \EXP{D}\prob{\widetilde{T}\geqslant T}.
\end{eqnarray*}
Note that $\EXP{D}$ is a function of the sampling rate $r$, and does not depend on
the detection policy. 

Consider the second term of Eqn. \ref{eq:min_del_delay_one_two_three},
\begin{eqnarray*}
\EXP{\left(\frac{K}{r}-T\right) I_{\{\widetilde{T}\ge T\}}}
& = & \EXP{\left(\frac{K}{r}-T\right) I_{\{\widetilde{K}\ge K\}}}\\
& = & \EXP{\left(\frac{K}{r}-T\right) I_{\{\widetilde{K}\ge K, S_0 = 1\}}}
    + \EXP{\left(\frac{K}{r}-T\right) I_{\{\widetilde{K}\ge K, S_0 = 0\}}}
\end{eqnarray*}
For $S_0 = 1$, we have $T = 0$ and $K = 0$. Hence, 
\begin{eqnarray*}
\EXP{\left(\frac{K}{r}-T\right) I_{\{\widetilde{T}\ge T\}}}
& = & 0 + {\mathsf E}_0\left[\left(\frac{K}{r}-T\right) I_{\{\widetilde{K}\ge K\}}\right]
\end{eqnarray*}
where ${\mathsf E}_0\left[\cdot\right]$ 
denote the expectation and ${\mathsf P}_0\left(\cdot\right)$ 
the probability law, when the 
initial state is $S_0 = 0$.
Now,
\begin{eqnarray}
\label{eqn:EZERO}
{\mathsf E}_0\left[\left(\frac{K}{r}-T\right) I_{\{\widetilde{K}\ge
K\}}\right]
& = & \sum_{b=1}^\infty\sum_{\widetilde{b}=b}^\infty \sum_{t = (b-1)/r+1}^{b/r} {\mathsf P}_0\left(
T=t, K=b, \widetilde{K}=\widetilde{b}\right)\cdot \left(\frac{b}{r}-t\right)\nonumber\\
& = & \sum_{b=1}^\infty\sum_{\widetilde{b}=b}^\infty 
{\mathsf P}_0\left(K=b, \widetilde{K} = \widetilde{b} \right)\nonumber\\
&&\hspace{-3mm}\cdot \left[\sum_{t = (b-1)/r+1}^{b/r} 
{\mathsf P}_0\left(T=t \mid K=b, \widetilde{K}=\widetilde{b}\right)
 \cdot \left(\frac{b}{r}-t\right)\right]
\end{eqnarray}

We note that $\widetilde{K}$ is independent of $T$ given $K$. Hence,
\begin{eqnarray*}
 {\mathsf E}_0\left[\left(\frac{K}{r}-T\right) I_{\{\widetilde{K}\ge K\}}\right]
& = & \sum_{b=1}^\infty\sum_{\widetilde{b}=b}^\infty 
{\mathsf P}_0\left(K=b, \widetilde{K} = \widetilde{b} \right)\\
& & \cdot\Big[\sum_{y = 0}^{1/r-1} 
y \cdot {\mathsf P}_0\left(T= \frac{b}{r} - y \mid K=b \right)\Big]
\end{eqnarray*}
We have
\begin{eqnarray*}
{\mathsf P}_0\left(T= t \mid K=b \right)  
& = & \left\{
      \begin{array}{ll}
	  \frac{(1-\rho)(1-p)^{t-1}p}{(1-\rho)(1-p_r)^{b-1}p_r}, & \mbox{ for $t$ s.t. } b =
	  \lceil t\cdot r \rceil \\
	  0, & \mbox{ otherwise.}
      \end{array}
      \right. 
\end{eqnarray*}

Hence, for $0 \le y \le 1/r -1$,
\begin{eqnarray*}
{\mathsf P}_0\left(T= \frac{b}{r} - y \mid K=b \right)  
&=&	  \frac{(1-p)^{b/r-y-1}p}{(1-p_r)^{b-1}p_r}
\end{eqnarray*}
But, $(1-p_r) = (1-p)^{1/r}$. Hence,
\begin{eqnarray*}
{\mathsf P}_0\left(T= \frac{b}{r} - y \mid K=b \right) 
&=&	  \frac{(1-p)^{b/r-y-1}p}{(1-p_r)^{b-1}p_r}\\
&=&	  \frac{(1-p)^{1/r-y-1}p}{1-(1-p)^{1/r}}
\end{eqnarray*}

It can be shown that 

\begin{eqnarray*}
\sum_{y=0}^{1/r-1}y\cdot \frac{(1-p)^{1/r-y-1}p}{1-(1-p)^{1/r}}
&=& \frac{1}{r} - \left(\frac{1}{p} - \frac{1}{rp_r}(1-p_r)\right)\\
&=:&l(r) 
\end{eqnarray*}

Therefore, Eqn. \ref{eqn:EZERO} can be written as
\begin{eqnarray*}
{\mathsf E}_0\left[\left(\frac{K}{r}-T\right) I_{\{\widetilde{K}\ge K\}}\right]
\ = \ l(r) \cdot {\mathsf P}_0\left(\widetilde{K}\ge K\right) 
& = & l(r) \cdot \left({\mathsf P}\left(\widetilde{K}\ge K\right) - \rho \right)\\
& = & l(r) \cdot \left(1 - {\mathsf P}\left(\widetilde{K} < K\right) - \rho \right)
\end{eqnarray*}

Finally, we have 
\begin{eqnarray*}
& & \min_{\Pi_\alpha} \EXP{(\widetilde{U}-T) I_{\{\widetilde{T}\ge
T\}}}\\ 
& = & \min_{\Pi_\alpha} \left\{d(r)\left(1 - {\mathsf P}\left(\widetilde{T} < T\right)\right)
+ l(r){\mathsf P}_0\left(\widetilde{T}\ge T\right)
 + \frac{1}{r}\EXP{(\widetilde{K}-K)^+} \right\}\\
& = & \min_{\Pi_\alpha} \left\{\left(d(r)+l(r)\right)\left(1 - {\mathsf P}\left(\widetilde{T} < T\right)\right) -\rho \cdot l(r) + \frac{1}{r}\EXP{(\widetilde{K}-K)^+}  \right\}
\end{eqnarray*}

Note that, in the above equation, the first term
$\left(d(r)+l(r)\right)\left(1 - {\mathsf P}\left(\widetilde{T} <
    T\right)\right)$ is minimum when ${\mathsf P}\left(\widetilde{T} <
  T\right) = \alpha$.
It follows that
\begin{eqnarray*}
\lefteqn{\min_{\Pi_\alpha} \EXP{(\widetilde{U}-T) I_{\{\widetilde{T}\ge T\}}}} \\
& \geqslant & \left(d(r)+l(r)\right)\left(1 - \alpha \right) - \rho\cdot l(r) 
+ \frac{1}{r}\min_{\Pi_\alpha}\EXP{(\widetilde{K}-K)^+}  \\
\end{eqnarray*}
Also, since the optimal policy for the problem
$\min_{\Pi_\alpha}\EXP{(\widetilde{K}-K)^+} $ achieves  $\left(1 -
  {\mathsf P}\left(\widetilde{T} < T\right)\right) = \alpha$, we also have 
\begin{eqnarray*}
  \left(d(r)+l(r)\right)\left(1 - \alpha \right) - \rho\cdot l(r) 
  + \frac{1}{r}\min_{\Pi_\alpha}\EXP{(\widetilde{K}-K)^+} &\geqslant& 
  \min_{\Pi_\alpha} \EXP{(\widetilde{U}-T) I_{\{\widetilde{T}\ge T\}}} 
\end{eqnarray*}
It follows that
\begin{eqnarray*}
\min_{\Pi_\alpha} \EXP{(\widetilde{U}-T) I_{\{\widetilde{T}\ge T\}}}
& = & \left(d(r)+l(r)\right)\left(1 - \alpha \right) - \rho\cdot l(r) 
+ \frac{1}{r}\min_{\Pi_\alpha}\EXP{(\widetilde{K}-K)^+}  \\
\end{eqnarray*}

We need $1 - \alpha > \rho$ or $\alpha < 1 - \rho$.
If $\alpha > 1 - \rho$, the optimal stopping is at $t = 0$. 
This will yield the desired probability of false alarm and 
$\EXP{(\widetilde{U}-T) I_{\{\widetilde{T}\ge T\}}} = 0$.
\qed

\section*{Appendix -- II}
\emph{Proof:} (Theorem~\ref{thm:fjq-gps_stationary_delay})
   The necessity of $Nr<\sigma$ is clear. The sufficiency
   proof goes as follows. Consider the FJQ-GPS system with every queue
   always containing a single dummy packet that is served at low
   priority. Let us call this the saturated
   FJQ-GPS system.  When a queue becomes empty, the low priority dummy
   packet contends for service. If it receives service, then it
   immediately reappears and continues to contend for service.  If,
   while a dummy packet is in service, a regular packet arrives, then
   the service of the dummy packet is preempted and the regular packet
   starts contending. It follows that the service rate applied to
   every queue (i.e., those with regular packets or those with dummy
   packets) is always $\sigma/N$. Now, consider a virtual service
   process of rate $\sigma$. In each slot, a service occurs with
   probability $\sigma$ and the service is applied to any one of the
   queues with equal probability.  Equivalently each queue is served
   by an independent Bernoulli process of rate $\sigma/N$.
   Considering only the services to the regular packets at each queue,
   we have a $GI/M/1$ queue (here $GI$ refers to a 
   {\em General distribution with Independent arrivals}, $M$ refers to 
   a {\em Markovian service} process and {\em 1 refers to one server}). 
   Hence, the system has proper stationary
   delay, iff $r < \sigma/N$. Also, it can be seen that the
   delays in the above described system (with dummy packets when a queue is empty) 
   upper bound those in the original
   FJQ-GPS system.  Hence, the result follows.
\qed

\section*{Appendix -- III}
\emph{Distribution of state noise ${\bf N}$}

Let ${\bf q} = [\lambda, b, \delta, {\bf w},{\bf r}]$.
Note that 
$\prob{M_k = m \big\arrowvert 
{\bf Q}_k ={\bf q},
{\bf \Theta}_k = {\mbox{\boldmath$\theta$}} 
} 
= \prob{M_k = m \big\arrowvert {\bf Q}_k = {\bf q}}$ and is given by 
\begin{align*}
\prob{M_k = 0 \big\arrowvert {\bf Q}_k = {\bf q}} = & \left\{
\begin{array}{ll}
1          & \hspace{4mm}\text{if} \ \phi_N({\bf q}) = 0\\ 
1-\sigma   & \hspace{4mm}\text{if} \ \phi_N({\bf q}) > 0 
\end{array}
\right. \\
\prob{M_k = m \big\arrowvert {\bf Q}_k = {\bf q}} = & \left\{
\begin{array}{ll}
0 & \ \ \ \text{if} \ \phi_N({\bf q}) = 0\\ 
\frac{\sigma}{\phi_N({\bf q})} & \ \ \  \text{if} \ \phi_{L^{(m)}}({\bf q}) > 0, \ \ m=1,2,3,\cdots,N. 
\end{array}
\right.
\end{align*}
where $\phi_N({\bf q})$ and $\phi_{L^{(m)}}({\bf q})$ are obtained from
Eqns.~\ref{eqn:phi-n} and \ref{eqn:phi-L-m}.

The distribution function,
$\prob{ O_{k} = o \big\arrowvert {\bf Q}_k = {\bf q},
{\bf \Theta}_k = {\mbox{\boldmath$\theta$}}}  =  
\prob{ O_{k} = o \big\arrowvert {\bf Q}_k = {\bf q},
{\Theta}_k = \theta }$ 
is given by
\begin{align*}
\prob{ O_{k} = o \big\arrowvert {\bf Q}_k = {\bf q},
{\Theta}_k = 0}  
= & \left\{
\begin{array}{ll}
1-p   & \ \ \ \text{if} \ o  = 0 \\
p     & \ \ \ \text{if} \ o  = 1,\\ 
0     & \ \ \ \text{otherwise}.
\end{array}
\right. \\ 
\prob{ O_{k} = o \big\arrowvert {\bf Q}_k = {\bf q},
{\Theta}_k = 1}  
= & \left\{
\begin{array}{ll}
1 & \hspace{7mm} \ \ \ \text{if} \ o  = 0\\ 
0 & \hspace{7mm} \ \ \ \text{otherwise}.
\end{array}
\right. 
\end{align*}

\section*{Appendix -- IV}
\emph{Proof of Lemma--1}

\vspace{2mm}

Let ${\bf q} = [\lambda, b, \delta, {\bf w},{\bf r}]$.
From Eqn.~\ref{eqn:pi},
{\footnotesize
\begin{eqnarray}
\Pi_k 
& := & \prob{ T \leqslant k \big\arrowvert {\bf I}_{k-1}, {\bf Q}_k =
{\bf q}, {\bf Y}_k  }\nonumber\\ 
&  = & \prob{ T \leqslant k-\delta   \big\arrowvert {\bf I}_{k-1}, {\bf
Q}_k = {\bf q},{\bf Y}_k  }  + \  
       \prob{ k-\delta < T \leqslant k \big\arrowvert {\bf I}_{k-1}, {\bf
	   Q}_k = {\bf q}, {\bf Y}_k } \nonumber\\ 
&  = & \prob{ T \leqslant k-\delta   \big\arrowvert {\bf I}_{k-1}, {\bf
Q}_k = {\bf q}, {\bf Y}_k }  + \  
       \prob{ T > k-\delta \big\arrowvert {\bf I}_{k-1}, {\bf Q}_k =
	   {\bf q}, {\bf Y}_k } \cdot 
       \prob{ T \leqslant k \big\arrowvert T > k-\delta,{\bf I}_{k-1},
	   {\bf Q}_k = {\bf q}, {\bf Y}_k },\nonumber\\ 
& = &  \Psi_k + (1-\Psi_k)\cdot\prob{ T \leqslant k \big\arrowvert T >
k-\delta,{\bf I}_{k-1}, {\bf Q}_k = {\bf q}, {\bf Y}_k },\nonumber\\ 
& = &\Psi_k  + (1-\Psi_k)\cdot \frac{\prob{ k - \delta < T \leqslant k}
\prob{ {\bf I}_{k-1}, {\bf Q}_k = {\bf q}, {\bf Y}_k  \big\arrowvert
k-\delta < T \leqslant k}} {\prob{ T > k - \delta } \prob{ {\bf
I}_{k-1}, {\bf Q}_k = {\bf q}, {\bf Y}_k  \big\arrowvert T > k-\delta }}\nonumber\\ 
\label{eqn:conditional-prob}
& = &{\Psi_k}  + (1-\Psi_k)\cdot \frac {\prob{ k - \delta < T \leqslant k} } {\prob{ T > k - \delta } }\\ 
\label{eqn:conditional-prob1}
& = & \hspace{-0mm} \Psi_k + (1-\Psi_k)\left(1 - (1-p)^{\delta}\right)
\end{eqnarray}
}
Eqn.~\ref{eqn:conditional-prob} is justified as follows.
Note that 
\begin{eqnarray*}
& & \prob{ {\bf I}_{k-1}, {\bf Q}_k = {\bf q}, {\bf Y}_k  \big\arrowvert k-\delta < T \leqslant k} \\ 
& = & \prob{ {\bf Q}_{[0:k-1]},  {\bf Q}_k = {\bf q},  {\bf X}_{[1:B_{k}-1]}, \{{X}^{(i)}_{B_{k}} : R_{k}^{(i)} =1\}, u_{[0:k-1]}\big\arrowvert k-\delta < T \leqslant k}\\
& = & \prob{ {\bf Q}_{[0:k-1]}, {\bf Q}_k = {\bf q}  \big\arrowvert
k-\delta < T \leqslant k}\nn
&  & \cdot
      \prob{ {\bf X}_{[1:B_{k}-1]}, \{{X}^{(i)}_{B_{k}} : R_{k}^{(i)}
	  =1\} \big\arrowvert k-\delta < T \leqslant k, {\bf Q}_{[0:k-1]}, {\bf Q}_k = {\bf q} }\nn
& &  \cdot 
      \prob{ u_{[0:k-1]} \big\arrowvert k-\delta < T \leqslant k,  {\bf
	  Q}_{[0:k-1]}, {\bf Q}_k = {\bf q}, {\bf X}_{[1:B_{k}-1]}, \{{X}^{(i)}_{B_{k}} : R_{k}^{(i)} =1\}}\\
& = & \prob{ {\bf Q}_{[0:k-1]}, {\bf Q}_k = {\bf q} }\cdot
      \prob{ {\bf X}_{[1:B_{k}-1]}, \{{X}^{(i)}_{B_{k}} : R_{k}^{(i)}
	  =1\} \big\arrowvert k-\delta < T, {\bf Q}_{[0:k-1]}, {\bf Q}_k = {\bf q} }\nn
& &  \cdot 
      \prob{ u_{[0:k-1]} \big\arrowvert {\bf Q}_{[0:k-1]}, {\bf Q}_k = {\bf q},  {\bf X}_{[1:B_{k}-1]}, \{{X}^{(i)}_{B_{k}} : R_{k}^{(i)} =1\}}\\
& = & \prob{ {\bf Q}_{[0:k-1]}, {\bf Q}_k = {\bf q}\big\arrowvert T > k-\delta  }\cdot
      \prob{ {\bf X}_{[1:B_{k}-1]}, \{{X}^{(i)}_{B_{k}} : R_{k}^{(i)}
	  =1\} \big\arrowvert T > k-\delta, {\bf Q}_{[0:k-1]},{\bf Q}_k = {\bf q} }\nn
& & \ \ \cdot 
      \prob{ u_{[0:k-1]} \big\arrowvert \big\arrowvert T > k-\delta ,
	  {\bf Q}_{[0:k-1]}, {\bf Q}_k = {\bf q},  {\bf X}_{[1:B_{k}-1]}, \{{X}^{(i)}_{B_{k}} : R_{k}^{(i)} =1\}}\\
& =  & \prob{ {\bf I}_{k-1}, {\bf Q}_k = {\bf q}, {\bf Y}_k  \big\arrowvert T > k-\delta }. 
\end{eqnarray*}
We use the following facts in the above justification: i) the 
evolution of the queueing system ${\bf Q}_k$ is independent of 
the change point $T$, ii) whenever $T > k-\delta$, the distribution 
of any sample $X^{(i)}_{h}$, $h \leqslant B_k$ is $f_0$, and 
iii) the control $u_k = 
\tilde{\mu}({\bf I}_k)$. 
\qed

\section*{Appendix -- V}
\emph{Recursive computation of $\Pi_k$}

At time $k$, based on the index of the node that successfully transmits
a packet $M_k$, the set of all sample paths $\Omega$ can be partitioned
based on the following events, 
\begin{eqnarray*}
\mathcal{E}_{1,k} & := & \left\{\omega: M_k(\omega)=0 \ \text{or} \
M_k(\omega) = j >0, R_k^{(j)}(\omega) = 1 \right\}\nn 
\mathcal{E}_{2,k} & := & \left\{\omega: 
M_k(\omega) = j >0, R_k^{(j)}(\omega) = 0, \sum_{i=1}^N
R_k^{(i)}(\omega) < N-1 \right\}\nn 
\mathcal{E}_{3,k} & := & \left\{\omega: 
M_k(\omega) = j >0, R_k^{(j)}(\omega) = 0, \sum_{i=1}^N
R_k^{(i)}(\omega) = N-1 \right\},
\end{eqnarray*}
i.e., $\Omega = \mathcal{E}_{1,k} \cup \mathcal{E}_{2,k} \cup \mathcal{E}_{3,k}$.
We note that the above events can also be described by using ${\bf Q}_k$
and ${\bf Q}_{k+1}$ in the following manner
\begin{eqnarray*}
\mathcal{E}_{1,k} & = & 
\left\{\omega: {\bf W}_{k+1}(\omega) = {\bf W}_k(\omega) , {\bf R}_{k+1}(\omega) = {\bf
R}_k(\omega) \right\} \\
& & \bigcup \left\{\omega:  {\bf W}_{k+1}(\omega) = {\bf W}_k(\omega) + {\bf e}_j, {\bf R}_{k+1}(\omega) = {\bf R}_k(\omega)\right\}\nn 
\mathcal{E}_{2,k} & = & \left\{\omega: 
{\bf W}_{k+1}(\omega) = {\bf W}_k(\omega), {\bf R}_{k+1}(\omega) = {\bf
R}_k(\omega) + {\bf e}_j \right\}\nn 
\mathcal{E}_{3,k} & = & \left\{\omega: 
\sum_{i=1}^N R_k^{(i)}(\omega) = N-1,
\forall i, W^{(i)}_{k+1}(\omega) = (W^{(i)}_k(\omega)-1)^+, 
R_{k+1}^{(i)}(\omega) = {\bf 1}_{\{W_k^{(i)}>0\}} 
\right\}.
\end{eqnarray*}

Here, the events $\mathcal{E}_{1,k}$ and $\mathcal{E}_{2,k}$ represent the
case $B_{k+1} = B_k$, and the event $\mathcal{E}_{3,k}$ represents the
case $B_{k+1} = B_k+1$ (i.e., only if the event $\mathcal{E}_{3,k}$
occurs then the batch index is incremented). We are interested in
obtaining $\Pi_{k+1}$ from $[{\bf Q}_k, {\Pi_k}]$ and ${\bf Z}_{k+1}$. 
We show that at time $k+1$, the statistic $\Psi_{k+1}$ (after having observed ${\bf
Z}_{k+1}$) can be computed in a recursive manner using $\Psi_k$ and
${\bf Q}_k$. Using Lemma~\ref{lemma} (using Eqn.~\ref{eqn:psi-theta}) we
compute $\Pi_{k+1}$ from $\Psi_{k+1}$.
\begin{eqnarray*}
\Psi_{k+1}
& = & \prob{\widetilde\Theta_{k+1}=1\mid {\bf I}_{k+1}} \\ 
& = & \sum_{c=1}^3 \prob{\widetilde\Theta_{k+1}=1, \mathcal{E}_{c,k}\mid
{\bf I}_{k+1}}\\  
& = & \sum_{c=1}^3 \prob{\widetilde\Theta_{k+1}=1 \mid \mathcal{E}_{c,k}, 
{\bf I}_{k+1}} {\bf 1}_{\mathcal{E}_{c,k}} \hspace{8mm} (\because \mathcal{E}_{c,k} \ \text{is} \
{\bf I}_{k+1} \ \text{measurable})
\end{eqnarray*}

\begin{itemize}
\item[$\bullet$] {\bf Case $M_k=0$ or $M_k = j>0$, $R_k^{(j)} = 1$}:
\begin{eqnarray*}
& & \Pi_{k+1}\nn 
& = & \prob{\Theta_{k+1}=1 \mid \mathcal{E}_{1,k}, {\bf I}_{k+1}}\nn 
& = & \prob{\Theta_{k+1}=1 \mid \mathcal{E}_{1,k}, {\bf I}_{k}, {\bf
Q}_{k+1} = {\bf q}'}\nn 
& = & \frac{\prob{\Theta_{k+1}=1\mid \mathcal{E}_{1,k}, {\bf I}_{k}} \cdot 
            f_{{\bf Q}_{k+1} \mid \Theta_{k+1}, \mathcal{E}_{1,k}, {\bf
			I}_{k}}({\bf q}'| 1, \mathcal{E}_{1,k}, {\bf I}_k)}
			{f_{{\bf Q}_{k+1}\big\arrowvert \mathcal{E}_{1,k}, {\bf
			I}_k}({\bf q}'|\mathcal{E}_{1,k}, {\bf I}_k)} 
\hspace{10mm}  (\text{by Bayes rule} ) \nn
& = & \prob{\Theta_{k+1}=1\mid \mathcal{E}_{1,k}, {\bf I}_{k}}
\hspace{35mm}  ({\bf
Q}_{k+1} \ \text{is independent of} \ \Theta_{k+1}) \nn
& = & \prob{\Theta_k =0, \Theta_{k+1}=1\mid {\bf I}_{k}} + \prob{\Theta_k =1, \Theta_{k+1}=1\mid {\bf I}_{k}}\\ 
& = & (1-\Pi_k)p + \Pi_k 
\end{eqnarray*}

\item[$\bullet$] {\bf Case $M_k = j>0$, $R_k^{(j)} = 0$,
$\sum_{i=1}^NR_k^{(i)} < N-1$}:
In this case, the sample $X_{B_k}^{(j)}$ is successfully transmitted and
is passed on to the decision maker. The decision maker receives just
this sample, and computes $\Pi_{k+1}$. We compute $\Psi_{k+1}$ from $\Psi_k$ 
and then we use Lemma~\ref{lemma} (using Eqn.~\ref{eqn:psi-theta}) to compute 
$\Pi_{k+1}$ from $\Psi_{k+1}$.
\begin{eqnarray*}
& & \Psi_{k+1}\nn
& = & \prob{\widetilde\Theta_{k+1}=1\mid {\mathcal E}_{2,k}, {\bf I}_{k+1}} \\ 
& = & \prob{\widetilde\Theta_{k+1}=1\mid {\mathcal E}_{2,k}, {\bf
I}_{k}, [{\mathbf Q}_{k+1}, {\bf Y}_{k+1}]=[{\bf q}',y]} \\ 
& = & \prob{\widetilde\Theta_k = 0, \widetilde\Theta_{k+1}=1\mid {\mathcal E}_{2,k}, {\bf
I}_{k}, [{\mathbf Q}_{k+1}, {\bf Y}_{k+1}]=[{\bf q}',y]}\\
& & 
+ \prob{\widetilde\Theta_k = 1, \widetilde\Theta_{k+1}=1\mid {\mathcal E}_{2,k}, {\bf
I}_{k}, [{\mathbf Q}_{k+1}, {\bf Y}_{k+1}]=[{\bf q}',y]} 
\end{eqnarray*}
Since, we consider the case when the fusion center received a sample at time $k+1$ and
$B_{k+1}=B_k$, $\Delta_{k+1} = \Delta_k+1$ and hence, the state $\widetilde\Theta_{k+1} =
\Theta_{k+1-\Delta_{k+1}} = \Theta_{k-\Delta_k} = \widetilde\Theta_k$.
Thus, in this case, $\Psi_{k+1}$ can be written as

{\footnotesize
\begin{eqnarray*}
& & \Psi_{k+1}\nn
& = & \prob{\widetilde\Theta_k = 1, \widetilde\Theta_{k+1}=1\mid {\mathcal E}_{2,k}, {\bf
I}_{k}, [{\mathbf Q}_{k+1}, {\bf Y}_{k+1}]=[{\bf q}',y]}\\ 
& 
\stackrel{(a)}{=} & \frac{ \prob{\widetilde\Theta_k = 1,
\widetilde\Theta_{k+1}=1\mid {\mathcal E}_{2,k},  {\bf I}_{k}} \cdot
\prob{{\bf Q}_{k+1} = {\bf q}' \mid \widetilde\Theta_k = 1,
\widetilde\Theta_{k+1}=1, {\mathcal E}_{2,k},  {\bf I}_{k}}}
{
{\mathsf P}({\bf Q}_{k+1}={\bf q}'|{\mathcal E}_{2,k}, {\bf I}_k)
\cdot f_{{\bf Y}_{k+1}\mid {\mathcal E}_{2,k}, {\bf I}_k, {\bf Q}_{k+1}}
(y|{\mathcal E}_{2,k}, {\bf I}_k,{\bf q}')
 } \nn
& & 
\cdot f_{{\bf Y}_{k+1}\mid \widetilde\Theta_k, \widetilde\Theta_{k+1},
{\mathcal E}_{2,k}, {\bf I}_k, {\bf Q}_{k+1} }
(y\mid 1, 1, {\mathcal E}_{2,k}, {\bf q}', {\bf I}_k )\nn 
& \stackrel{(b)}{=} & \frac{ \prob{\widetilde\Theta_k = 1, \widetilde\Theta_{k+1}=1\mid {\mathcal E}_{2,k}, {\bf I}_{k}} \cdot
{\mathsf P}({\bf Q}_{k+1}={\bf q}'|{\mathcal E}_{2,k}, {\bf I}_k)
\cdot f_{{\bf Y}_{k+1}\mid \widetilde\Theta_k}
(y\mid 1 ) }{
{\mathsf P}({\bf Q}_{k+1}={\bf q}'|{\mathcal E}_{2,k}, {\bf I}_k)
\cdot   f_{{\bf Y}_{k+1}\mid {\mathcal E}_{2,k}, {\bf I}_k, {\bf Q}_{k+1}}
(y|{\mathcal E}_{2,k}, {\bf I}_k,{\bf q}')
 } \nn
& \stackrel{(c)}{=} & \frac{ \prob{\widetilde\Theta_k = 1, \widetilde\Theta_{k+1}=1\mid {\mathcal E}_{2,k}, {\bf I}_{k}} 
\cdot f_1(y)}{
\prob{\widetilde\Theta_k=0 \mid {\mathcal E}_{2,k}, {\bf I}_k, {\bf
Q}_{k+1}}\cdot f_{{\bf Y}_{k+1}\mid \widetilde\Theta_k
}(y|0)
+
\prob{\widetilde\Theta_k=1 \mid {\mathcal E}_{2,k}, {\bf I}_k, {\bf
Q}_{k+1}}\cdot f_{{\bf Y}_{k+1}\mid \widetilde\Theta_k
}(y|1)
} \nn
& \stackrel{(d)}{=} & \frac{ \Psi_k f_1(y) } { (1-\Psi_k)  f_0(y) +\Psi_k f_1(y) } 
\end{eqnarray*}
}
We explain the steps $(a), (b), (c), (d)$ below.
\begin{itemize}
\item[(a)] By Bayes rule, for events $A,B,C,D,E,F$, we have 
 \[\prob{A B \mid C D E F} = 
           \frac{\prob{A B \mid C D} \prob{E \mid A B C D} \prob{F \mid A
		   B C D E }}{\prob{E\mid C D}\prob{F \mid C D E}}\] 
\item[(b)] ${\bf Q}_{k+1}$ is independent of $\widetilde\Theta_k$,
$\widetilde\Theta_{k+1}$. Also, given $\widetilde\Theta_k$, 
${\bf Y}_{k+1}$ is independent of \ \ \ $\widetilde\Theta_{k+1},
{\mathcal E}_{2,k}, {\bf I}_k, {\bf Q}_{k+1}$

\item[(c)] For any events $A,B$, and a continuous random variable $Y$,
the conditional density
           function $f_{Y|A}(y|A) = \prob{B \mid A}  f_{Y|AB}(y|AB) + \prob{B^c \mid A}
		   f_{Y|A B^c}(y|A B^c)$. Also,  given $\widetilde\Theta_k$, 
${\bf Y}_{k+1}$ is independent of 
${\mathcal E}_{2,k}, {\bf I}_k, {\bf Q}_{k+1}$

\item[(d)]  ${\mathcal E}_{2,k}$ is $[{\bf I}_{k}, {\bf Q}_{k+1}]$
measurable, and hence, given  $[{\bf I}_{k}, {\bf Q}_{k+1}]$, 
$\widetilde\Theta_k$ is independent of ${\mathcal E}_{2,k}$.
\end{itemize}

\item[$\bullet$] {\bf Case $M_k = j>0$, $R_k^{(j)} = 0$,
$\sum_{i=1}^NR_k^{(i)} = N-1$}:
In this case, at time $k+1$, the decision maker receives the last sample
of batch $B_k$, $X_{B_k}^{(j)}$ (that is successfully transmitted during
slot $k$) and the samples of batch $B_k+1$, if any, that are queued in
the sequencer buffer. We compute $\Psi_{k+1}$ from $\Psi_k$ and then we
use Lemma~\ref{lemma} (using Eqn.~\ref{eqn:psi-theta}) to compute
$\Pi_{k+1}$ from $\Psi_{k+1}$. In this case, the decision maker 
receives $n := \sum_{i=1}^N {\bf 1}_{\{W_k^{(i)}>0\}}$ samples of batch
$B_k+1$. Also, note that $n$ is ${\bf I}_k$ measurable.

\begin{eqnarray*}
 \Psi_{k+1}
& = & \prob{\widetilde\Theta_{k+1}=1\mid {\mathcal E}_{3,k}, {\bf I}_{k+1}} \\ 
& = & \prob{\widetilde\Theta_{k+1}=1\mid {\mathcal E}_{3,k}, {\bf
I}_{k}, [{\mathbf Q}_{k+1}, {\bf Y}_{k+1}]=[{\bf q}',{\bf y}]} \\ 
& = & \prob{\widetilde\Theta_k = 0, \widetilde\Theta_{k+1}=1\mid {\mathcal E}_{3,k}, {\bf
I}_{k}, [{\mathbf Q}_{k+1}, {\bf Y}_{k+1}]=[{\bf q}',{\bf y}]}\\
& & 
+ \prob{\widetilde\Theta_k = 1, \widetilde\Theta_{k+1}=1\mid {\mathcal E}_{3,k}, {\bf
I}_{k}, [{\mathbf Q}_{k+1}, {\bf Y}_{k+1}]=[{\bf q}',{\bf y}]}
\end{eqnarray*}
Since, we consider the case $B_{k+1}=B_k+1$, $\Delta_{k+1} =
\Delta_k+1-1/r$ and hence, the state $\widetilde\Theta_{k+1} =
\Theta_{k+1-\Delta_{k+1}} = \Theta_{k-\Delta_k+1/r}$.

Let ${\bf y} = [y_0,y_1,\cdots,y_n]$. Consider
{\footnotesize
\begin{eqnarray*}
&  & \prob{\widetilde\Theta_k = \widetilde\theta, \widetilde\Theta_{k+1}=1\mid {\mathcal E}_{3,k}, {\bf
I}_{k}, [{\mathbf Q}_{k+1}, {\bf Y}_{k+1}]=[{\bf q}',{\bf y}]}\\ 
& 
\stackrel{(a)}{=} & \frac{ \prob{\widetilde\Theta_k = \widetilde\theta, \widetilde\Theta_{k+1}=1\mid {\mathcal E}_{3,k},  {\bf I}_{k}} \cdot \prob{{\bf Q}_{k+1} = {\bf q}' \mid \widetilde\Theta_k = \widetilde\theta, \widetilde\Theta_{k+1}=1, {\mathcal E}_{3,k},  {\bf I}_{k}}
}{{\mathsf P}({\bf Q}_{k+1}={\bf q}'|{\mathcal E}_{3,k}, {\bf I}_k)
\cdot f_{{\bf Y}_{k+1}\mid {\mathcal E}_{3,k}, {\bf I}_k, {\bf Q}_{k+1}}
({\bf y}|{\mathcal E}_{3,k}, {\bf I}_k,{\bf q}')
 } \nn
& & \cdot
f_{{\bf Y}_{k+1}\mid \widetilde\Theta_k, \widetilde\Theta_{k+1},
{\mathcal E}_{3,k}, {\bf I}_k, {\bf Q}_{k+1} }({\bf y}\mid
\widetilde\theta, 1, {\mathcal E}_{3,k}, {\bf q}', {\bf I}_k )\nn 
& \stackrel{(b)}{=} & \frac{ \prob{\widetilde\Theta_k = \widetilde\theta, \widetilde\Theta_{k+1}=1\mid {\mathcal E}_{3,k}, {\bf I}_{k}} \cdot
{\mathsf P}({\bf Q}_{k+1}={\bf q}'|{\mathcal E}_{3,k}, {\bf I}_k)
\cdot f_{\widetilde\theta}(y_0) \prod_{i=1}^n f_1(y_i)
}{
{\mathsf P}({\bf Q}_{k+1}={\bf q}'|{\mathcal E}_{3,k}, {\bf I}_k)
\cdot   f_{{\bf Y}_{k+1}\mid {\mathcal E}_{3,k}, {\bf I}_k, {\bf Q}_{k+1}}
({\bf y}|{\mathcal E}_{3,k}, {\bf I}_k,{\bf q}')
 } \nn
& \stackrel{(c)}{=} & \frac{ \prob{\widetilde\Theta_k = \widetilde\theta
\mid {\mathcal E}_{3,k}, {\bf I}_{k}} \cdot \prob{\widetilde\Theta_{k+1} = 1
\mid \widetilde\Theta_k = \widetilde\theta, {\mathcal E}_{3,k}, {\bf I}_{k}}  
\cdot f_{\widetilde\theta}(y_0) \prod_{i=1}^n f_1(y_i)}{
\sum_{\widetilde\theta'=0}^1 
\sum_{\widetilde\theta''=0}^1 
\prob{\widetilde\Theta_k= \widetilde\theta', 
\widetilde\Theta_{k+1}= \widetilde\theta'', 
\mid {\mathcal E}_{3,k}, {\bf I}_k, {\bf
Q}_{k+1}}\cdot 
f_{{\bf Y}_{k+1}\mid \widetilde\Theta_k,\widetilde\Theta_{k+1} {\mathcal
E}_{3,k}, {\bf I}_k, {\bf
Q}_{k+1}}(y|\widetilde\theta',\widetilde\theta'',{\mathcal E}_{3,k},
{\bf I}_k,{\bf q}')
}. 
\end{eqnarray*}
}
We explain the steps $(a), (b), (c)$ below.
\begin{itemize}
\item[(a)] By Bayes rule, for events $A,B,C,D,E,F$, we have 
 \[\prob{A B \mid C D E F} = 
           \frac{\prob{A B \mid C D} \prob{E \mid A B C D} \prob{F \mid A
		   B C D E }}{\prob{E\mid C D}\prob{F \mid C D E}}\] 
\item[(b)] ${\bf Q}_{k+1}$ is independent of $\widetilde\Theta_k$, 
$\widetilde\Theta_{k+1}$. Also, given $\widetilde\Theta_k$, 
${\bf Y}_{k+1,0}$ is independent of  $\widetilde\Theta_{k+1},
{\mathcal E}_{3,k}, {\bf I}_k, {\bf Q}_{k+1}$, and 
given $\widetilde\Theta_{k+1}$, 
${\bf Y}_{k+1,i}$ is independent of  $\widetilde\Theta_{k},
{\mathcal E}_{3,k}, {\bf I}_k, {\bf Q}_{k+1}$. It is to be noted that 
given the state of nature, the sensor measurements $Y_{k+1,0},Y_{k+1,1},\cdots,Y_{k+1,n}$
are conditionally independent.

\item[(c)] For any events $A,B$, and a continuous random variable $Y$,
the conditional density
           function $f_{Y|A}(y|A) = \prob{B \mid A}  f_{Y|AB}(y|AB) + \prob{B^c \mid A}
		   f_{Y|A B^c}(y|A B^c)$. Also,  given $\widetilde\Theta_k$, 
${\bf Y}_{k+1}$ is independent of 
${\mathcal E}_{3,k}, {\bf I}_k, {\bf Q}_{k+1}$
\end{itemize}
It is to be noted that the event ${\mathcal E}_{3,k}$ is $[{\bf I}_{k}, {\bf Q}_{k+1}]$
measurable, and hence, given  $[{\bf I}_{k}, {\bf Q}_{k+1}]$, 
$\widetilde\Theta_k$ is independent of ${\mathcal E}_{3,k}$.
Thus, in this case,
{\footnotesize
\begin{eqnarray*}
\Psi_{k+1} 
& = &
\frac{(1-\Psi_k)p_r f_0(y_0) \prod_{i=1}^n f_1(y_i)
+ \Psi_k f_1(y_0)\prod_{i=1}^n
f_1(y_i)}{
(1-\Psi_k)(1-p_r) f_0(y_0) \prod_{i=1}^n f_0(y_i)+ 
(1-\Psi_k) p_r    f_0(y_0) \prod_{i=1}^n f_1(y_i) + 
    \Psi_k         f_1(y_0)\prod_{i=1}^n f_1(y_i)}.
\end{eqnarray*}
}

Thus, 
using Lemma~\ref{lemma} (using Eqn.~\ref{eqn:psi-theta}), we have 
\begin{eqnarray*}
\Pi_{k+1} & = & \Psi_{k+1} + (1-\Psi_{k+1})(1 - (1-p)^{\Delta_{k+1}})
\nn
 & =: & \phi_{\Psi}(\Psi_{k},{\bf Z}_{k+1}) +
 \left(1-\phi_{\Psi}(\Psi_{k},{\bf Z}_{k+1})\right) (1 -
 (1-p)^{\Delta_{k+1}})\nn
 & = & \phi_{\Psi}\left(\frac{\Pi_{k}-
 (1-(1-p)^{\Delta_{k}})}{(1-p)^{\Delta_{k}}},{\bf Z}_{k+1}\right)\nn
 & &+
 \left(1-
  \phi_{\Psi}\left(\frac{\Pi_{k}-
 (1-(1-p)^{\Delta_{k}})}{(1-p)^{\Delta_{k}}},{\bf Z}_{k+1}\right) 
 \right) (1 -
 (1-p)^{\Delta_{k+1}})\nn
 & =: & \phi_{\Pi}\left([{\bf Q}_k,\Pi_k],{\bf Z}_{k+1}\right).
\end{eqnarray*}

\end{itemize}

\section*{Appendix -- VI}
\emph{Structure of $\tau^*$}
We use the following Lemma to show that $J^*({\bf q},\pi)$ is concave in
$\pi$. 
\begin{lemma}\label{Lemma01}
If $f:[0,1] \to \mathbb{R}$ is concave, then the function 
$h:[0,1] \to \mathbb{R}$ defined by 
\begin{align*} 
h(y) & = {\mathsf E}_{\phi({\bf x})}\left[
f\left(
\frac{y \cdot \phi_2({\bf x}) + (1-y)p_r \cdot \phi_1({\bf x})}
{y \cdot\phi_2({\bf x}) + (1-y)p_r\cdot\phi_1({\bf x}) + (1-y)(1-p_r)\cdot\phi_0({\bf x})}\right) \right]   
\end{align*}
is concave for each ${\bf x}$, where $\phi({\bf x}) = y \cdot\phi_2({\bf x}) + (1-y)p_r\cdot\phi_1({\bf x}) 
+ (1-y)(1-p_r)\cdot\phi_0({\bf x})$, $0 < p_r < 1$,
and  
 $\phi_0({\bf x})$, $\phi_1({\bf x})$, and $\phi_2({\bf x})$ are pdfs on ${\bf X}$.
\end{lemma}
\begin{proof}
See Appendix -- I of \cite{premkumar-kumar08sleep-wake-scheduling}.
\end{proof}
Note that in the finite $H$--horizon (truncated version of Eqn.~\ref{eqn:need-to-solve}), 
we note from {\em value iteration} that the cost--to--go function, 
for a given ${\bf q}$, $J_H^H([{\bf q},\pi]) = 1 - \pi$ is concave in $\pi$. 
Hence, by  Lemma~\ref{Lemma01}, we see that for any given  ${\bf q}$,
the cost--to--go functions  $J_{H-1}^H([{\bf q},\pi]),$ $J_{H-2}^H([{\bf q},\pi]),$ 
$\cdots, J_0^H([{\bf q},\pi])$ are concave in $\pi$.
Hence for $0 \le \lambda \le 1$, 
\begin{align*} 
J^*([{\bf q},\pi]) & = \lim_{H \rightarrow \infty} J^H_0([{\bf q},\pi])\\
J^*([{\bf q},\lambda \pi_1 + (1-\lambda)\pi_2]) & = \lim_{H \rightarrow \infty} J^H_0\Big([{\bf q},\lambda\pi_1+(1-\lambda)\pi_2]\Big)\\
& \ge \lim_{H \rightarrow \infty} \lambda J^H_0([{\bf q},\pi_1])+ \lim_{H \rightarrow \infty} (1-\lambda)J^H_0([{\bf q},\pi_2])\\
& = \lambda J^*([{\bf q},\pi_1])+(1-\lambda)J^*([{\bf q},\pi_2])
\end{align*}
It follows that for any given ${\bf q}$, $J^*([{\bf q},\pi])$ is concave in $\pi$.
$\hfill\qed$

Define the map $\xi:\mathcal{Q}\times[0,1] \to \mathbb{R}_+$ 
as $\xi([{\bf q},\pi]) := 1 - \pi$ and the map 
$\kappa:\mathcal{Q}\times[0,1] \to \mathbb{R}_+$, as 
$\kappa([{\bf q},\pi]) := c\cdot\pi + A_{J^*}([{\bf q},\pi]) 
  =   c\cdot\pi + {\sf E}\left[J^*\left([{\bf Q}_{k+1},
\phi_{\Pi}(\nu_k,{\bf Z}_{k+1})]\right)\bigg\arrowvert \nu_k = [{\bf
q},\pi]\right]$.
Note that $\xi([{\bf q},1]) =  0$, $\kappa([{\bf q}, 1]) =  c$, 
$\xi([{\bf q},0]) =   1$ and 
\begin{eqnarray*} 
\kappa([{\bf q},0]) 
 & = &  {\sf E}\left[J^*\left([{\bf Q}_{k+1}, \phi_{\Pi}(\nu_k,{\bf Z}_{k+1})]\right)\bigg\arrowvert \nu_k = [{\bf q},0]\right] \\
 & \stackrel{(2)}{=} &  {\sf E}\left[J^*\left([\phi_{\bf Q}({\bf Q}_{k},M_k), \phi_{\Pi}(\nu_k,{\bf Z}_{k+1})]\right)\bigg\arrowvert \nu_k = [{\bf q},0]\right] \\
 & = &  \sum_{m=0}^N {\sf E}\left[J^*\left([\phi_{\bf Q}({\bf q},m), \phi_{\Pi}(\nu_k,{\bf Z}_{k+1})]\right)\bigg\arrowvert M_k = m, \nu_k = [{\bf q},0]\right] \prob{M_k = m \bigg\arrowvert \nu_k = [{\bf q},0]}\\
 & \stackrel{(4)}{\leqslant} &  \sum_{m=0}^N J^*\left(\left[\phi_{\bf Q}({\bf q},m), 
{\sf E}\left[\phi_{\Pi}(\nu_k,{\bf Z}_{k+1})
\bigg\arrowvert M_k = m, \nu_k = [{\bf q},0]\right]
\right]\right)
 \prob{M_k = m \bigg\arrowvert \nu_k = [{\bf q},0]}\\
 & = &  \sum_{m=0}^N J^*\left([\phi_{\bf Q}({\bf q},m), p\right) \prob{M_k = m \bigg\arrowvert \nu_k = [{\bf q},0]}\\
 & \stackrel{(6)}{\leqslant} &  \sum_{m=0}^N \left(1-p\right) \cdot \prob{M_k = m \bigg\arrowvert \nu_k = [{\bf q},0]}\\
& = & 1 - p  \ < \ 1
\end{eqnarray*}
where in the above derivation, 
we use the evolution of ${\bf Q}_k$ in step 2, the Jensen's inequality (as
for any given ${\bf q}$, $J^*({\bf q},\pi)$ is concave in $\pi$) in step 4, and 
the inequality $J^*({\bf q},\pi) \leqslant 1-\pi$ in step 6.

Note that  $\kappa([{\bf q}, 1]) - \xi([{\bf q},1]) > 0$ and 
$\kappa([{\bf q}, 0]) - \xi([{\bf q},0]) <  0$. Also, for a fixed ${\bf q}$, 
the function $\kappa([{\bf q},\pi]) - \xi([{\bf q}, \pi])$ is concave in $\pi$.
Hence, by the {\em intermediate value theorem}, for a fixed ${\bf q}$, 
there exists $\gamma({\bf q}) \in [0,1]$ such that $\kappa([{\bf q},\gamma]) = 
\xi([{\bf q},\gamma])$.
This $\gamma$ is unique as $\kappa([{\bf q},\pi]) = \xi([{\bf q},\pi])$ for at most two values of $\pi$. 
If in the interval $[0,1]$, there are two distinct values of $\pi$ for 
which $\kappa([{\bf q},\pi]) = \xi([{\bf q}, \pi])$, then the signs of 
$\kappa([{\bf q},0]) - \xi([{\bf q}, 0])$ and
$\kappa([{\bf q},1]) - \xi([{\bf q}, 1])$
should be the same.
Hence, 
\begin{align*} 
\tau^* & = \inf\left\{k: \Pi_k \geqslant \gamma({\bf Q}_k)\right\}
\end{align*}
where the threshold $\gamma({\bf q})$ is given by
$c\cdot\gamma({\bf q}) + A_{J^*}([{\bf q},\gamma({\bf q})]) = 1 -
\gamma({\bf q})$. \qed

\bibliographystyle{acmtrans} 
\bibliography{premkumar-etal10optimising-seq-det}

\begin{received}
Received March 2009;
revised December 2009 and June 2010; accepted Month Year
\end{received}

\end{document}